\documentclass[lettersize,journal]{IEEEtran}
\usepackage{amsthm}
\usepackage{amsmath,amsfonts}
\usepackage{algorithmic}
\usepackage{algorithm}
\usepackage{amssymb}
\usepackage{array}
\usepackage[caption=false,font=normalsize,labelfont=sf,textfont=sf]{subfig}
\usepackage{textcomp}
\usepackage{stfloats}
\usepackage{url}
\usepackage{verbatim}
\usepackage{graphicx}
\usepackage{cite}
\usepackage{cases}
\usepackage{hyperref}
\usepackage{bm}
\usepackage{color}
\usepackage{multirow}
\usepackage{booktabs}
\usepackage{tabularx}

\hypersetup{colorlinks=true,
	linkcolor=blue,
	citecolor=blue,      
	urlcolor=black,
}

\newtheoremstyle{mylemma}
{}{}                 
{\normalfont}        
{}                   
{\bfseries}         
{.}                  
{ }                  
{\thmname{#1}\thmnumber{ #2}\thmnote{ (#3)}} 

\theoremstyle{mylemma}
\newtheorem{lemma}{Lemma}
\newtheorem{remark}{Remark}

\newtheorem{proposition}{Proposition}
\newtheorem{definition}{Definition}

\begin{document}

\title{Fidelity Where it Matters: Site-Specific Nonuniform Refinement for Wireless Digital Twins}

\author{Zihao Zhou,~\IEEEmembership{Graduate Student Member,~IEEE}, Zhaolin Wang,~\IEEEmembership{Member,~IEEE}, \\ and Yuanwei Liu,~\IEEEmembership{Fellow,~IEEE}
\thanks{The authors are with the Department of Electrical and Computer Engineering, The University of Hong Kong, Hong Kong (e-mail: eezihaozhou@connect.hku.hk,zhaolin.wang@hku.hk,yuanwei@hku.hk)}}

\maketitle

\begin{abstract}
Wireless digital twins (WDTs) enable site-specific learning, management, and evaluation, but constructing and maintaining uniformly high-fidelity WDTs for large-scale urban environments is costly. This paper studies \underline{t}ask-\underline{o}riented \underline{n}onuniform WDT \underline{r}efinement (TONR) by addressing the following question: \textit{given a limited sensing and reconstruction budget, which building geometries in the initial WDT should be refined to best preserve wireless fidelity?} A resource-constrained building selection problem is formulated to minimize the expected discrepancy between the wireless response of the refined WDT and the physical environment. Through local first-order approximation, the refinement value of each building is shown to depend jointly on its correctable geometry uncertainty and the sensitivity of the task response to that uncertainty. This value is then estimated solely from the initial low-fidelity WDT using physically interpretable geometry perturbations and central finite differences. To reduce the computational cost, a propagation-relevance ellipsoid is proposed to filter out buildings unlikely to contribute significantly to the wireless propagation. The resulting knapsack problem is solved under both equal and heterogeneous refinement costs. Simulations across multiple urban scenarios show that the proposed algorithm can substantially improve wireless fidelity by refining only a small subset of buildings.
\end{abstract}


\section{Introduction}
\IEEEPARstart{D}{igital} twins (DTs), in a broad sense, refer to synchronized digital replicas of physical entities or systems that accurately reflect the behavior, performance, and operating conditions of their real-world counterparts~\cite{Tao2024wireless}. In recent years, digital twins have been widely applied in various domains, such as smart cities~\cite{Lv2022deep}, healthcare~\cite{Chen2023digital}, and industrial automation~\cite{Gehrmann2020digital}. For wireless communications, recent advances in sensing, three-dimensional (3D) reconstruction, and ray-tracing have made wireless digital twins (WDTs)~\cite{Alkhateeb2023real, Khan2022digital} promising tools for site-specific analysis, management, and optimization~\cite{Wang2026generative, Zhou2026beam}. More specifically, the value of WDTs mainly lies in the following aspects.
\begin{itemize}
	\item \textbf{\emph{From site-agnostic to site-specific:}} Many existing works rely on simplified analytical, stochastic, or parameterized channel models that are not explicitly tied to the specific deployment environment. By incorporating physical scene and propagation mechanisms, WDTs provide a platform for site-specific wireless modeling and optimization, thereby improving the reliability of wireless system design.
	
	\item \textbf{\emph{Data source for model training:}} Learning-based wireless algorithms usually require large-scale datasets, which are difficult and costly to collect in the real world. WDTs can provide synthetic yet site-specific data for training and evaluating such algorithms, thereby reducing the dependence on extensive real-world data collection~\cite{Jia2025urban, Wang2025deeptelecom}.
	
	\item \textbf{\emph{Safe and controllable validation:}} WDTs provide a sandbox-like virtual environment where communication algorithms and protocols can be safely validated and compared without disrupting live network operation. In particular, multiple what-if scenarios can be examined in a repeatable and controllable manner before real-world deployment~\cite{Lin20236G}.
	
	\item \textbf{\emph{Closed-loop optimization:}} The optimized configurations obtained in WDT can be transferred to the physical network through a digital-physical interface~\cite{IETF2022digital}. Meanwhile, real-time data interaction between the physical wireless network and its corresponding WDT also enables a closed-loop optimization process~\cite{Fang2025sensing}.
\end{itemize}

Despite these advantages, the practical value of WDTs critically depends on the \emph{fidelity} of their underlying digital representation, including, for example, 3D geometry, material properties, ray-tracing configurations, and hardware modeling~\cite{Luo2025digital}. In this context, fidelity should not be interpreted merely as visual or geometric realism, but rather as the ability of a WDT to preserve task-relevant wireless behaviors. Since signal propagation is highly sensitive to the surrounding environment, insufficient fidelity, resulting from imperfect geometric reconstruction, inaccurate material assignment, or parameter mismatch, may undermine the reliability of downstream tasks such as radio-map estimation, channel characterization and beamforming optimization. A straightforward solution is to uniformly refine all components to sufficiently high fidelity. However, constructing and maintaining such a high-fidelity WDT over a large-scale environment is highly resource-intensive, involving substantial costs in sensing, reconstruction, storage, updating, and simulation. This challenge becomes even more severe in complex urban scenarios under a limited resource budget. As a result, a key question naturally arises: \emph{is wireless fidelity equally affected by all environmental components represented in the WDT?} If not, \emph{how should limited refinement resources be judiciously allocated across these components to preserve wireless fidelity as much as possible?} Nevertheless, existing WDT construction pipelines typically treat fidelity as a global design choice, while the task-dependent and component-level allocation of refinement resources remains unexplored.

\subsection{Related Work}
WDTs are becoming increasingly important in next-generation wireless communication systems, where they serve as site-specific platforms for analysis, optimization, and evaluation. Existing studies on WDT-assisted wireless networks can be broadly divided into three categories: 1) the construction of WDTs~\cite{Alkhateeb2019deepmimo, Jia2025urban, Wang2025deeptelecom, Yu2024channelGPT, Gong2025digital, Zhang2025denoising, Jiang2023digital}, 2) the exploitation of WDTs for network optimization and decision-making~\cite{Cai2025digital, Belgiovine2026better, Karakelle2026ray, Alikhani2024digital, Zhang2026efficient, Fang2025optimizing, Shui2023cell, Cui2023digital, Ho2025ai}, and 3) the refinement and calibration of WDTs for enhancing wireless fidelity~\cite{Luo2025digital,Ruah2024calibrating,Luo2026wireless}.

\subsubsection{The construction of WDTs}
Since WDTs are expected to provide site-specific digital representations of physical wireless environments, a key research direction is to construct WDTs and generate reliable data from them. As an early effort, DeepMIMO, a generic and parameterized deep-learning dataset, was released in~\cite{Alkhateeb2019deepmimo}, with its channel data obtained via ray tracing (RT) from Remcom Wireless InSite~\cite{Remcom}. UrbanMIMOMap further extended this direction to dense urban radio-map learning by providing multi-input multi-output (MIMO) channel state information (CSI) maps and precoding-aware benchmarks~\cite{Jia2025urban}. More recently, the authors of~\cite{Wang2025deeptelecom} developed DeepTelecom, a WDT dataset generation pipeline, which constructs the third level of detail (LoD3) indoor and outdoor scenes with material-parameterizable surfaces and produces multi-modal data. Beyond ray tracing-based data generation, recent studies have also explored generative AI (GenAI) for digital-twin channel generation. In~\cite{Yu2024channelGPT}, an LLM-driven DT channel generator named ChannelGPT was introduced, which exploits multimodal data and environmental information to generate multi-scenario channel parameters. The authors of~\cite{Gong2025digital} studied the construction of a DT of channel (DToC), where the positions of user terminals are regarded as physical objects and the corresponding CSI data are modeled as virtual digital objects. A conditional diffusion model was adopted to learn the relationship between physical and virtual digital objects, enabling the generation of user-specific statistical CSI. Similarly, in~\cite{Zhang2025denoising}, a diffusion model-based digital-twinning framework was developed for generating integrated sensing and communication (ISAC) MIMO channels. The practical value of WDT-generated data has also been validated in downstream learning tasks. In~\cite{Jiang2023digital}, the CSI dataset generated by WDT was used to train a neural network (NN) for beam prediction. The results showed that the NN trained solely on WDT-generated data can still achieve promising performance on real-world measurements, highlighting the potential of WDTs to reduce the dependence on costly real-world data collection.

\subsubsection{The exploitation of WDTs}
A DT-assisted CSI prediction framework was proposed in~\cite{Cai2025digital} to reduce the CSI acquisition overhead. This framework extracts an environment-specific channel subspace basis from a WDT as an environment prior, and fuses it with partial real-time CSI to recover the full spatial-frequency channel. In~\cite{Belgiovine2026better}, a WDT based on NVIDIA Sionna and Aerial Omniverse Digital Twin (AODT) was constructed to jointly optimize the positions of unmanned aerial vehicles (UAVs), antenna orientations, and transmit powers for wireless coverage. The authors of~\cite{Karakelle2026ray} constructed a DT by leveraging NVIDIA Sionna based on a 3D map, and then used this DT to optimize the phase configuration of a reconfigurable intelligent surface (RIS). With the help of the DT, the extensive signaling overhead associated with estimating high-dimensional RIS channels can be eliminated. Also focusing on DT-assisted RIS communications, in~\cite{Alikhani2024digital}, a DT was adopted for beamforming and interference management. Furthermore, a DT-specific robust transmission scheme was designed by considering the channel approximation errors from the DT. Beyond providing site-specific information to assist wireless network optimization, WDTs can also serve as virtual environments for training reinforcement learning (RL) agents. In~\cite{Zhang2026efficient}, a conditional generative adversarial network (cGAN)-based DT module was constructed to generate virtual state-action transition pairs. An offline DRL-based policy was trained on this DT for efficient beam selection in a cell-free ISAC system. Focusing on the codebook design problem, the DRL agent subsequently explored optimal beam patterns efficiently in the DT~\cite{Fang2025optimizing}, thereby avoiding extensive real-world interactions. The concept of leveraging DT as a virtual training environment for RL can also be found in~\cite{Shui2023cell, Cui2023digital, Ho2025ai}.

\subsubsection{The refinement and calibration of WDTs}
The third research direction shifts the focus from constructing or exploiting
WDTs to improving the fidelity of the WDT itself, since the effectiveness of a WDT critically depends on the consistency between the digital representation and the physical wireless environment. In~\cite{Luo2025digital}, instead of treating the WDT as a perfect data generator, the authors explicitly investigated how mismatches in geometry, material properties, ray-tracing configuration, and hardware modeling affect the reliability of DT-generated CSI data and the performance of the downstream CSI feedback task. In~\cite{Ruah2024calibrating}, an approach based on the variational expectation maximization algorithm was proposed for ray tracing calibration. In~\cite{Luo2026wireless}, a learning-based calibration framework was proposed to refine the DFT-domain channel information generated by a low-complexity WDT. Although these studies highlight the importance of WDT fidelity and provide effective calibration mechanisms, they mainly focus on either analyzing global fidelity factors or correcting the generated channel information after the WDT has been constructed. In contrast, the task-oriented and resource-constrained refinement of the WDT itself remains largely unexplored. Existing WDT calibration studies mainly focus on calibrating model parameters or correcting the generated wireless responses. WDT refinement, by contrast, generally aims to improve the fidelity of the underlying digital representation, such as its geometric model. However, existing WDT construction and refinement pipelines typically treat fidelity as a scene-level design choice and do not differentiate environmental components according to their relevance to wireless tasks. A critical question therefore remains unexplored: \textit{under a limited refinement budget, how should resources be judiciously allocated across different WDT components so as to maximize task-level wireless fidelity?}

\subsection{Contributions}
Motivated by the aforementioned gap, this work investigates task-oriented nonuniform refinement for WDTs, where limited resources are selectively allocated to the WDT components that matter most to downstream wireless tasks. Specifically, the contributions of this paper are summarized as follows:
\begin{itemize}
	\item We formulate a task-oriented and resource-constrained WDT refinement problem at the building level. Instead of pursuing geometric accuracy uniformly over the entire scene, the proposed formulation selects buildings to minimize the expected discrepancy between the wireless responses of the refined WDT and the physical environment. By locally linearizing the task response and modeling the reduction in residual geometry uncertainty, we derive a building-level refinement value that jointly captures the correctable geometry error and its impact on the downstream wireless task.
	
	\item We develop a practical refinement method that uses only the initial low-fidelity WDT. To estimate the refinement value, we construct a computable surrogate for the refinement value using structured geometric perturbations and central-difference estimates of task-response sensitivity. To further reduce the computational cost, a propagation-relevance ellipsoid is constructed to screen buildings according to their potential involvement in both line-of-sight (LoS) and non-line-of-sight (NLoS) propagation. The resulting budget-constrained selection problem is solved as a knapsack problem. Moreover, we establish an end-to-end optimality-gap bound that separately characterizes the errors due to response linearization, structured directional approximation, finite-difference estimation, and candidate screening.
	
	\item We evaluate the proposed algorithm in realistic 3D scenarios in Hong Kong, considering both equal and heterogeneous refinement costs. Numerical results show that by refining only a small fraction of buildings, the proposed algorithm consistently achieves the highest accuracy of radio map prediction under equal refinement costs, and provides clear performance gains when a moderate heterogeneous-cost budget is available. Ablation studies further characterize the effects of the main design parameters. Overall, these results confirm that task-oriented nonuniform refinement (TONR) can provide a cost-effective alternative to full-scene high-fidelity reconstruction for WDTs.
\end{itemize}

The remainder of this paper is organized as follows. Section \ref{sec:system_model} presents the system model and formulates the task-oriented and resource-constrained WDT refinement problem. Section \ref{sec:proposed_method} introduces the proposed building-level refinement method and provides its optimality gap and computational complexity analyses. Section~\ref{sec:numerical_results} presents the numerical results and evaluates the proposed method under different settings. Finally, Section \ref{sec:conclusion} concludes the paper.

\textit{Notation:}
Scalars, vectors, and matrices are denoted by italic letters, bold lowercase letters, and bold uppercase letters, respectively. Sets are denoted by calligraphic letters, and $|\mathcal{X}|$ denotes the cardinality of a set $\mathcal{X}$. The real and complex Euclidean spaces of dimension $n$ are denoted by $\mathbb{R}^{n}$ and $\mathbb{C}^{n}$, respectively. The transpose and conjugate transpose are denoted by $(\cdot)^{\mathsf{T}}$ and $(\cdot)^{\mathsf{H}}$, respectively. The magnitude of a scalar and the Euclidean norm of a vector are denoted by $|\cdot|$ and $\|\cdot\|_2$, respectively. The Frobenius norm of a matrix is denoted by $\Vert\cdot\Vert_{\mathrm F}$. Moreover, $\mathbb{E}[\cdot]$, $\Re\{\cdot\}$, and $\operatorname{tr}(\cdot)$ denote expectation, the real part, and the matrix trace, respectively, while $\operatorname{Area}(\cdot)$ and $\operatorname{Vol}(\cdot)$ denote the area and volume of a geometric region. For Hermitian matrices $\mathbf{A}$ and $\mathbf{B}$, $\mathbf{A}\succeq\mathbf{B}$ means that $\mathbf{A}-\mathbf{B}$ is positive semidefinite. 

\section{System Model and Problem Formulation}\label{sec:system_model}
This section first introduces the system model for a WDT and then formulates the building-level task-oriented WDT refinement problem. We consider an existing low-fidelity WDT and determine how a limited refinement budget should be allocated among its building components. The refinement objective is defined by the performance of the wireless task in the physical environment, rather than by a purely geometric similarity measure.

Specifically, consider a physical environment
\begin{equation}\label{eq:physical_env}
	\setlength\abovedisplayskip{3pt}
	\setlength\belowdisplayskip{3pt}
	\mathcal{E}^{\star} = (\mathcal{S}^\star_{{\rm fix}}, \{\mathcal{G}_i^\star, \mathcal{M}_i^\star\}_{i\in \mathcal{I}}),
\end{equation}
where the superscript $\star$ denotes a physical quantity. The set $\mathcal{I}=\{1, \cdots, N_b\}$ contains $N_b$ building components. For building $i\in \mathcal{I}$, $\mathcal{G}_i^\star$ and $\mathcal{M}_i^\star$ represent its physical geometry and electromagnetic material properties, respectively. $\mathcal{S}^\star_{{\rm fix}}$ collects the remaining environmental elements, such as terrain, vegetation, and infrastructure, that are not refined in this work.

The initial WDT can be written as
\begin{equation}
	\setlength\abovedisplayskip{3pt}
	\setlength\belowdisplayskip{3pt}
	\mathcal{D}_0 = (\hat{\mathcal{E}}_0, \hat{\mathcal{P}}, \boldsymbol{\Theta}),
\end{equation}
where $\hat{\mathcal{E}}_0$ is the initial digital representation of the physical environment in (\ref{eq:physical_env}), $\hat{\mathcal{P}}$ denotes the adopted wireless propagation model, which may include LoS, reflection, diffraction, scattering, etc. $\boldsymbol{\Theta}$ collects the numerical configurations of the propagation solver. For example, $\hat{\mathcal{P}}$ can be instantiated using the existing RT tools such as NVIDIA Sionna RT~\cite{Aoudia2025sionna} and Remcom Wireless InSite~\cite{Remcom}, while $\boldsymbol{\Theta}$ may include the maximum interaction depth, sampling density, and other solver-related parameters. In this work, $\hat{\mathcal{P}}$ and $\boldsymbol{\Theta}$ remain fixed during refinement. Following the definition in (\ref{eq:physical_env}), $\hat{\mathcal{E}}_0$ can be further expressed as
\begin{equation}
	\setlength\abovedisplayskip{3pt}
	\setlength\belowdisplayskip{3pt}
	\hat{\mathcal{E}}_0 = (\hat{\mathcal{S}}_{{\rm fix}}, \{\hat{\mathcal{G}}_i^{(0)}, \hat{\mathcal{M}}_i\}_{i\in \mathcal{I}}),
\end{equation}
where $\hat{\mathcal{S}}_{{\rm fix}}$ is the digital representation of the environmental elements that will not be refined, $\hat{\mathcal{G}}_i^{(0)}$ is the initial digital geometry of the $i$-th building\footnote{For example, it may denote the set of mesh faces or the set of point samples representing building $i$.}, and $\hat{\mathcal{M}}_i$ is its digital material representation. Because the initial WDT is typically constructed under limited sensing resources, e.g., with a low point-cloud density, $\hat{\mathcal{G}}_i^{(0)}$ may differ substantially from the physical geometry $\mathcal{G}_i^\star$.

Let $\eta = (\mathcal{X}_T, \mathcal{X}_R, \boldsymbol{\xi})$ denote the configuration of the wireless system, where $\mathcal{X}_T\subseteq \mathbb{R}^{3\times 1}$ and $\mathcal{X}_R\subseteq \mathbb{R}^{3\times 1}$ are the sets of transmitter (Tx) and receiver (Rx) locations, respectively. $\boldsymbol{\xi}$ contains the remaining system parameters, including the carrier frequency, antenna configurations, and bandwidth. Therefore, under configuration $\eta$, the task-relevant wireless response generated by the initial WDT is
\begin{equation}\label{eq:wireless_response}
	\setlength\abovedisplayskip{3pt}
	\setlength\belowdisplayskip{3pt}
	\hat{\bm{{\rm y}}}_0(\eta) = \mathcal{F}_{\mathcal{T}}(\mathcal{D}_0, \eta),
\end{equation}
where $\mathcal{T}$ denotes the wireless task, e.g., radio map prediction, and $\mathcal{F}_{\mathcal{T}}(\cdot)$ is the corresponding WDT-based response operator. The output $\hat{\bm{{\rm y}}}_0(\eta)$ can represent, for example, the path-gain samples of a radio map or the channel coefficients.

In this work, building refinement is implemented by increasing the density of point-cloud measurements collected from the physical environment. Let $z_i\in\{0, 1\}$ denote the refinement decision for building $i\in \mathcal{I}$, where $z_i=0$ means retaining its initial digital geometry and $z_i=1$ means applying the prescribed refinement operation. Thus, the joint refinement decision can be represented by
\begin{equation}
	\setlength\abovedisplayskip{3pt}
	\setlength\belowdisplayskip{3pt}
	\bm{{\rm z}} = (z_1, \cdots, z_{N_b}) \in \{0, 1\}^{N_b}.
\end{equation}

For each building $i\in \mathcal{I}$, recall that $\hat{\mathcal{G}}_i^{(0)}$ represents its initial low-fidelity geometry. Let  $\hat{\mathcal{G}}_i^{(1)}$ denote the geometry obtained after applying the prescribed refinement. The geometry included in the refined WDT is thus given by
\begin{equation}
	\setlength\abovedisplayskip{3pt}
	\setlength\belowdisplayskip{3pt}
	\hat{\mathcal{G}}_i^{(z_i)} = \left\{
	\begin{aligned}
		\hat{\mathcal{G}}_i^{(0)} &, z_i = 0, \\
		\hat{\mathcal{G}}_i^{(1)} &, z_i = 1, \\
	\end{aligned}
	\right.
\end{equation}
For a given refinement decision $\bm{{\rm z}}$, the resulting digital environment is
\begin{equation}
	\setlength\abovedisplayskip{3pt}
	\setlength\belowdisplayskip{3pt}
	\hat{\mathcal{E}}_{\bm{{\rm z}}} = (\hat{\mathcal{S}}_{\rm fix}, \{\hat{\mathcal{G}}_i^{(z_i)}, \hat{\mathcal{M}}_i \}_{i\in \mathcal{I}}).
\end{equation}
The corresponding refined WDT can be represented as
\begin{equation}
	\setlength\abovedisplayskip{3pt}
	\setlength\belowdisplayskip{3pt}
	\mathcal{D}_{\bm{{\rm z}}} = (\hat{\mathcal{E}}_{\bm{{\rm z}}}, \hat{\mathcal{P}}, \boldsymbol{\Theta}).
\end{equation}
Consequently, the task-relevant wireless response generated by the refined WDT is
\begin{equation}\label{eq:task-relevant_wireless_response}
	\setlength\abovedisplayskip{3pt}
	\setlength\belowdisplayskip{3pt}
	\hat{\bm{{\rm y}}}_{\bm{{\rm z}}}(\eta) = \mathcal{F}_{\mathcal{T}}(\mathcal{D}_{\bm{{\rm z}}}, \eta).
\end{equation}
Let $\bm{{\rm y}}^\star(\eta)$ denote the corresponding wireless response in the physical environment under $\eta$. We measure the task-oriented discrepancy between the refined WDT and the physical environment using the weighted squared response error
\begin{equation}\label{eq:discrepancy}
	\setlength\abovedisplayskip{3pt}
	\setlength\belowdisplayskip{3pt}
	\mathcal{L}_{\mathcal{T}}(\hat{\bm{{\rm y}}}_{\bm{{\rm z}}}(\eta), \bm{{\rm y}}^\star(\eta)) = \Vert \hat{\bm{{\rm y}}}_{\bm{{\rm z}}}(\eta) - \bm{{\rm y}}^\star(\eta) \Vert^2_{\bm{{\rm W}}},
\end{equation}
where $\Vert\bm{{\rm u}} \Vert^2_{\bm{{\rm W}}} \triangleq \bm{{\rm u}}^{{\rm H}}\bm{{\rm W}} \bm{{\rm u}}$, and $\bm{{\rm W}} \succeq \bm{{\rm 0}}$ is a weighting matrix. It should be noted that this discrepancy in (\ref{eq:discrepancy}) generally contains both geometry-induced errors and residual errors caused by material mismatch, propagation model mismatch, numerical solver configurations, and hardware imperfections. Since the refinement action considered in this work modifies only the building geometry, the geometry-dependent component of this discrepancy will be isolated in Section \ref{sec:proposed_method}.

Let $c_i > 0$ denote the refinement cost of applying the prescribed refinement operation to building $i$. In our work, this cost may be measured or estimated by the number of additionally acquired points. Under an additive cost model, the total refinement cost can be expressed as 
\begin{equation}
	\setlength\abovedisplayskip{3pt}
	\setlength\belowdisplayskip{3pt}
	C(\bm{{\rm z}}) = \sum_{i=1}^{N_b}c_iz_i.
\end{equation}

Therefore, the proposed building-level task-oriented WDT refinement problem can be formulated as
\begin{subequations}
	\setlength\abovedisplayskip{3pt}
	\setlength\belowdisplayskip{3pt}
	\begin{align}
		(\textbf{P1}): \min_{\bm{{\rm z}}} \quad & \mathbb{E}_{\eta}\left[\Vert \hat{\bm{{\rm y}}}_{\bm{{\rm z}}}(\eta) - \bm{{\rm y}}^\star(\eta) \Vert^2_{\bm{{\rm W}}}\right] \\
		\mathrm{s.t.} \quad & z_i \in \{0, 1\}, i\in \mathcal{I}, \label{constraint1}\\
		& C(\bm{{\rm z}}) = \sum_{i=1}^{N_b}c_iz_i \leq B.  \label{constraint2}
	\end{align}
\end{subequations}
The objective function in (\textbf{P1}) measures the expected discrepancy in task response between the refined WDT and the physical environment. However, (\textbf{P1}) cannot be solved directly because evaluating its objective function requires both the physical response $\bm{{\rm y}}^\star(\eta)$ and the refined WDT response $\hat{\bm{{\rm y}}}_{\bm{{\rm z}}}(\eta)$ for each possible decision $\bm{{\rm z}}$. The former is generally unavailable without extensive physical measurements, while obtaining the latter would require a prohibitively large number of RT evaluations since $N_b$ binary decisions give rise to an exponentially growing search space. These challenges motivate the practical building-level refinement method developed in the following sections.

\section{The Proposed TONR Method}\label{sec:proposed_method}
In this section, we first transform the original intractable problem (\textbf{P1}) into a tractable surrogate problem. We then derive an optimality-gap bound and analyze the computational complexity.
\subsection{Task-Oriented Building Refinement Value}
Recall that $\hat{\mathcal{G}}_i^{(0)}$ and $\hat{\mathcal{G}}_i^{(1)}$ denote the initial and refined geometric models for building $i$, respectively. For local analysis, we represent the initial, refined, and physical geometries of building $i$ in a common finite-dimensional coordinate system. Let $\phi_i$ denote a geometry parameterization that maps each of these geometries to a vector in $\mathbb{R}^{d_i\times 1}$, where $d_i$ is the number of geometry parameters. Thus, we define $\bm{{\rm \gamma}}_i(z_i) \triangleq \phi_i(\hat{\mathcal{G}}_i^{(z_i)}) \in \mathbb{R}^{d_i \times 1}$. Since $\mathcal{G}_i^\star$ is the unknown physical geometry of building $i$, let $\bm{{\rm \gamma}}_i^\star \triangleq \phi_i(\mathcal{G}_i^\star)$ be its representation in the same coordinate system. Based on this definition, the residual geometry error after taking action $z_i$ is defined as
\begin{equation}
	\setlength\abovedisplayskip{3pt}
	\setlength\belowdisplayskip{3pt}
	\bm{{\rm e}}_i(z_i) \triangleq \bm{{\rm \gamma}}_i(z_i) - \bm{{\rm \gamma}}_i^\star \in \mathbb{R}^{d_i \times 1}.
\end{equation}
Collecting the geometry representations of all buildings gives
\begin{equation}
	\setlength\abovedisplayskip{3pt}
	\setlength\belowdisplayskip{3pt}
	\bm{{\rm \gamma}}(\bm{{\rm z}}) = [\bm{{\rm \gamma}}_1^{\mathsf{T}}(z_1), \cdots, \bm{{\rm \gamma}}_{N_b}^{\mathsf{T}}(z_{N_b})]^{\mathsf{T}},
\end{equation}
with the physical geometry representations given by
\begin{equation}
	\setlength\abovedisplayskip{3pt}
	\setlength\belowdisplayskip{3pt}
	\bm{{\rm \gamma}}^\star = [(\bm{{\rm \gamma}}_1^\star)^{\mathsf{T}}, \cdots, (\bm{{\rm \gamma}}_{N_b}^\star)^{\mathsf{T}}]^{\mathsf{T}}.
\end{equation}
Let $\hat{\bm{{\rm y}}}(\bm{{\rm \gamma}}^\star; \eta)$ denote the response produced by the adopted WDT propagation model when the physical building geometries are used. In general, $\hat{\bm{{\rm y}}}(\bm{{\rm \gamma}}^\star; \eta)$ may not be equivalent to $\bm{{\rm y}}^\star(\eta)$, because the material representation, propagation model, solver configuration, and hardware model may remain imperfect. We therefore define the residual non-geometric mismatch as
\begin{equation}\label{eq:b}
	\setlength\abovedisplayskip{3pt}
	\setlength\belowdisplayskip{3pt}
	\bm{{\rm b}}(\eta) \triangleq \hat{\bm{{\rm y}}}(\bm{{\rm \gamma}}^\star; \eta) - \bm{{\rm y}}^\star(\eta).
\end{equation}
Since these non-geometric WDT components remain fixed during the considered refinement procedure, $\bm{{\rm b}}(\eta)$ does not depend directly on the building refinement decision $\bm{{\rm z}}$.
Consequently, the task-relevant wireless response defined in (\ref{eq:task-relevant_wireless_response}) can be written as $\hat{\bm{{\rm y}}}_{\bm{{\rm z}}}(\eta) = \hat{\bm{{\rm y}}}(\bm{{\rm \gamma}}(\bm{{\rm z}}); \eta)$. We locally approximate this response around the initial geometry $\bm{{\rm \gamma}}(\bm{{\rm 0}})$. The first-order expansion of the response of the refined WDT is
\begin{equation}\label{eq:taylor1}
	\setlength\abovedisplayskip{3pt}
	\setlength\belowdisplayskip{3pt}
	\hat{\bm{{\rm y}}}(\bm{{\rm \gamma}}(\bm{{\rm z}}); \eta) = \hat{\bm{{\rm y}}}(\bm{{\rm \gamma}}(\bm{{\rm 0}}); \eta) + \sum_{i\in \mathcal{I}}\bm{{\rm J}}_i(\eta)\left[\bm{{\rm \gamma}}_i(z_i) - \bm{{\rm \gamma}}_i(0)\right] + \bm{{\rm r}}_{\bm{{\rm z}}}(\eta),
\end{equation}
where 
\begin{equation}
	\setlength\abovedisplayskip{3pt}
	\setlength\belowdisplayskip{3pt}
	\bm{{\rm J}}_i(\eta) \triangleq \left.\frac{\partial \hat{\bm{{\rm y}}}(\bm{{\rm \gamma}}; \eta)}{\partial \bm{{\rm \gamma}}_i^{\mathsf{T}}}\right|_{\bm{{\rm \gamma}} = \bm{{\rm \gamma}}(\bm{{\rm 0}})  }
\end{equation}
is the response Jacobian with respect to the parameterized geometry of building $i$, evaluated at the initial WDT, and $\bm{{\rm r}}_{\bm{{\rm z}}}(\eta)$ is the corresponding higher-order remainder. Similarly, expanding the WDT response associated with the physical building geometries around the same point gives
\begin{equation}\label{eq:taylor2}
	\setlength\abovedisplayskip{3pt}
	\setlength\belowdisplayskip{3pt}
	\hat{\bm{{\rm y}}}(\bm{{\rm \gamma}}^\star; \eta) = \hat{\bm{{\rm y}}}(\bm{{\rm \gamma}}(\bm{{\rm 0}}); \eta) + \sum_{i\in \mathcal{I}}\bm{{\rm J}}_i(\eta)\left[\bm{{\rm \gamma}}_i^\star - \bm{{\rm \gamma}}_i(0)\right] + \bm{{\rm r}}^\star(\eta).
\end{equation}
Subtracting (\ref{eq:taylor2}) from (\ref{eq:taylor1}) yields
\begin{equation}\label{eq:taylor_response_error}
	\setlength\abovedisplayskip{3pt}
	\setlength\belowdisplayskip{3pt}
	\hat{\bm{{\rm y}}}(\bm{{\rm \gamma}}(\bm{{\rm z}}); \eta) - \hat{\bm{{\rm y}}}(\bm{{\rm \gamma}}^\star; \eta) = \sum_{i\in \mathcal{I}}\bm{{\rm J}}_i(\eta)\bm{{\rm e}}_i(z_i) + \bm{{\rm r}}_{\bm{{\rm z}}}(\eta) - \bm{{\rm r}}^\star(\eta).
\end{equation}
Combining (\ref{eq:b}) and (\ref{eq:taylor_response_error}), the refined WDT response error relative to the physical environment is 
\begin{equation}\label{eq:exact_response_decoposition}
	\setlength\abovedisplayskip{3pt}
	\setlength\belowdisplayskip{3pt}
	\hat{\bm{{\rm y}}}(\bm{{\rm \gamma}}(\bm{{\rm z}}); \eta) - \bm{{\rm y}}^\star(\eta) \!=\! \sum_{i\in \mathcal{I}}\bm{{\rm J}}_i(\eta)\bm{{\rm e}}_i(z_i) + \bm{{\rm r}}_{\bm{{\rm z}}}(\eta) - \bm{{\rm r}}^\star(\eta) + \bm{{\rm b}}(\eta).
\end{equation}
Hence, neglecting the higher-order terms gives
\begin{equation}\label{eq:first_order_approx}
	\setlength\abovedisplayskip{3pt}
	\setlength\belowdisplayskip{3pt}
	\hat{\bm{{\rm y}}}(\bm{{\rm \gamma}}(\bm{{\rm z}}); \eta) - \bm{{\rm y}}^\star(\eta) \approx \sum_{i\in \mathcal{I}}\bm{{\rm J}}_i(\eta)\bm{{\rm e}}_i(z_i) + \bm{{\rm b}}(\eta).
\end{equation}

Since $\mathcal{G}^\star_i$ and $\hat{\mathcal{G}}_i^{(1)}$ are unavailable before refinement, we model $\bm{{\rm e}}_i(z_i)$ as an uncertain quantity conditioned on the initial low-fidelity WDT. Define its second-moment matrix as
\begin{equation}
	\setlength\abovedisplayskip{3pt}
	\setlength\belowdisplayskip{3pt}
	\boldsymbol{\Sigma}_i(z_i) \triangleq \mathbb{E}[\bm{{\rm e}}_i(z_i) \bm{{\rm e}}_i^{\mathsf{H}}(z_i)] \succeq \bm{{\rm 0}}.
\end{equation}
We assume that the prescribed refinement operation is second-moment non-increasing, i.e., $\boldsymbol{\Sigma}_i(1) \preceq \boldsymbol{\Sigma}_i(0)$. Under this monotonicity assumption, the reduction in geometry-error second moment is defined as 
\begin{equation}\label{eq:uncertainty_reduction}
	\setlength\abovedisplayskip{3pt}
	\setlength\belowdisplayskip{3pt}
	\Delta \boldsymbol{\Sigma}_i \triangleq \boldsymbol{\Sigma}_i(0) - \boldsymbol{\Sigma}_i(1) \succeq \bm{{\rm 0}},
\end{equation}
leading to $\boldsymbol{\Sigma}_i(z_i) = \boldsymbol{\Sigma}_i(0) - z_i\Delta \boldsymbol{\Sigma}_i$. Based on the first-order approximation in (\ref{eq:first_order_approx}) and the residual non-geometric mismatch defined in (\ref{eq:b}), we define the surrogate objective of (\textbf{P1}) as
\begin{equation}\label{eq:surrogate1}
	\setlength\abovedisplayskip{3pt}
	\setlength\belowdisplayskip{3pt}
	\widetilde{R}(\bm{{\rm z}}) \triangleq \mathbb{E}_{\eta, \bm{{\rm e}}}\left[\Big\Vert \sum_{i\in \mathcal{I}} \bm{{\rm J}}_i(\eta)\bm{{\rm e}}_i(z_i) + \bm{{\rm b}}(\eta) \Big\Vert^2_{\bm{{\rm W}}} \right].
\end{equation}
Expanding the weighted squared norm in (\ref{eq:surrogate1}) gives
\begin{equation}\label{eq:decomposition}
	\setlength\abovedisplayskip{3pt}
	\setlength\belowdisplayskip{3pt}
	\begin{aligned}
		\widetilde{R}(\bm{{\rm z}})&=\mathbb{E}\left[\Vert \bm{{\rm b}}(\eta) \Vert^2_{\bm{{\rm W}}}\right] \\
		&+\sum_{i\in \mathcal{I}}\mathbb{E}_{\eta, \bm{{\rm e}}}\left[\bm{{\rm e}}_i^{\mathsf{H}}(z_i) \bm{{\rm J}}_i^{\mathsf{H}}(\eta) \bm{{\rm W}} \bm{{\rm J}}_i(\eta)\bm{{\rm e}}_i(z_i)\right] \\
		&+\sum_{i\in\mathcal{I}}\sum_{\substack{j\in\mathcal{I}\\j\neq i}}\mathbb{E}_{\eta, \bm{{\rm e}}}\left[\bm{{\rm e}}_i^{\mathsf{H}}(z_i) \bm{{\rm J}}_i^{\mathsf{H}}(\eta) \bm{{\rm W}} \bm{{\rm J}}_j(\eta) \bm{{\rm e}}_j(z_j)   \right] \\
		&+2\Re\left\{\sum_{i\in \mathcal{I}} \mathbb{E}_{\eta, \bm{{\rm e}}}\left[\bm{{\rm b}}^{\mathsf{H}}(\eta)\bm{{\rm W}}\bm{{\rm J}}_i(\eta)\bm{{\rm e}}_i(z_i) \right]\right\}.
	\end{aligned}
\end{equation}
For analytical tractability, we make the following orthogonality assumptions. First, the residual geometry errors of different buildings satisfy
\begin{equation}\label{eq:assumption1}
	\setlength\abovedisplayskip{3pt}
	\setlength\belowdisplayskip{3pt}
	\mathbb{E}_{\bm{{\rm e}}}[\bm{{\rm e}}_i(z_i) \bm{{\rm e}}_j^{\mathsf{H}}(z_j)] = \bm{{\rm 0}}, \forall i\neq j, i\in\mathcal{I}, j\in\mathcal{I}.
\end{equation}
Second, the residual non-geometric mismatch $\bm{{\rm b}}(\eta)$ is assumed to be orthogonal, under the task-weighted inner product, to the first-order geometry-induced response error of each building:
\begin{equation}\label{eq:assumption2}
	\setlength\abovedisplayskip{3pt}
	\setlength\belowdisplayskip{3pt}
	\mathbb{E}_{\eta, \bm{{\rm e}}}\left[\bm{{\rm b}}^{\mathsf{H}}(\eta)\bm{{\rm W}}\bm{{\rm J}}_i(\eta)\bm{{\rm e}}_i(z_i) \right]=0, \forall i\in\mathcal{I}, z_i\in\{0,1\}.
\end{equation}
For (\ref{eq:decomposition}), we first let $\bm{{\rm A}}_{ij}(\eta) = \bm{{\rm J}}_i^{\mathsf{H}}(\eta) \bm{{\rm W}} \bm{{\rm J}}_j(\eta)$ to simplify notation. Since $\bm{{\rm e}}_i^{\mathsf{H}}(z_i)\bm{{\rm A}}_{ij}(\eta)\bm{{\rm e}}_j(z_j)$ is a scalar, we have\footnote{It is assumed that the geometry error statistics are independent of $\eta$.}
\begin{equation}
	\setlength\abovedisplayskip{3pt}
	\setlength\belowdisplayskip{3pt}
	\begin{aligned}
		\setlength\abovedisplayskip{3pt}
		\setlength\belowdisplayskip{3pt}
		& \mathbb{E}_{\eta, \bm{{\rm e}}}\left[\bm{{\rm e}}_i^{\mathsf{H}}(z_i) \bm{{\rm J}}_i^{\mathsf{H}}(\eta) \bm{{\rm W}} \bm{{\rm J}}_j(\eta) \bm{{\rm e}}_j(z_j)   \right] \\
		& = \mathbb{E}_{\eta}\left[\mathbb{E}_{\bm{{\rm e}}}\left[{\rm tr}\left(\bm{{\rm e}}_i^{\mathsf{H}}(z_i) \bm{{\rm A}}_{ij}(\eta) \bm{{\rm e}}_j(z_j) \right)\right] \right] \\
		& = \mathbb{E}_{\eta}\left[{\rm tr}\left[\mathbb{E}_{\bm{{\rm e}}}\left(\bm{{\rm e}}_i^{\mathsf{H}}(z_i) \bm{{\rm A}}_{ij}(\eta) \bm{{\rm e}}_j(z_j)\right)\right]\right] \\
		& = \mathbb{E}_{\eta}\left[{\rm tr}\left[\bm{{\rm A}}_{ij}(\eta)\mathbb{E}_{\bm{{\rm e}}}\left(\bm{{\rm e}}_j(z_j) \bm{{\rm e}}_i^{\mathsf{H}}(z_i)\right)\right]\right]=0.
	\end{aligned}
\end{equation}
Consequently, under assumptions (\ref{eq:assumption1}) and (\ref{eq:assumption2}), we obtain that 
\begin{equation}\label{eq:decomposition2}
	\setlength\abovedisplayskip{3pt}
	\setlength\belowdisplayskip{3pt}
	\begin{aligned}
	\widetilde{R}(\bm{{\rm z}}) &= R_b + \sum_{i\in \mathcal{I}} \mathbb{E}_{\eta}\left[\mathbb{E}_{\bm{{\rm e}}}\left[\bm{{\rm e}}_i^{\mathsf{H}}(z_i) \bm{{\rm J}}_i^{\mathsf{H}}(\eta) \bm{{\rm W}} \bm{{\rm J}}_i(\eta) \bm{{\rm e}}_i(z_i)\right]\right] \\
	&=R_b + \sum_{i\in \mathcal{I}}\mathbb{E}_{\eta}\left[{\rm tr}\left(\bm{{\rm J}}_i^{\mathsf{H}}(\eta) \bm{{\rm W}} \bm{{\rm J}}_i(\eta) \mathbb{E}_{\bm{{\rm e}}}\left[\bm{{\rm e}}_i(z_i)\bm{{\rm e}}^{\mathsf{H}}_i(z_i)  \right]\right)\right]    \\
	&=R_b + \sum_{i\in \mathcal{I}}\mathbb{E}_{\eta}\left[{\rm tr}\left(\bm{{\rm J}}_i^{\mathsf{H}}(\eta) \bm{{\rm W}} \bm{{\rm J}}_i(\eta) \boldsymbol{\Sigma}_i(z_i)\right)\right],
\end{aligned}
\end{equation}
where $R_b\triangleq \mathbb{E}_{\eta}\left[\Vert \bm{{\rm b}}(\eta) \Vert^2_{\bm{{\rm W}}}\right]$. Substituting (\ref{eq:uncertainty_reduction}) into (\ref{eq:decomposition2}) gives
\begin{equation}\label{eq:rz}
	\setlength\abovedisplayskip{3pt}
	\setlength\belowdisplayskip{3pt}
	\widetilde{R}(\bm{{\rm z}}) = \widetilde{R}(\bm{{\rm 0}}) - \sum_{i\in \mathcal{I}}z_i v_i,
\end{equation}
where $v_i \triangleq \mathbb{E}_{\eta}[{\rm tr}(\bm{{\rm J}}_i^{\mathsf{H}}(\eta) \bm{{\rm W}} \bm{{\rm J}}_i(\eta) \Delta \boldsymbol{\Sigma}_i  )]$ is the approximate task-risk reduction obtained by refining building $i$, which can also be interpreted as the refinement value of building $i$.

The variable $v_i$ provides an important insight: \textit{the refinement value of a building depends jointly on the geometry error reduction via refinement, and the sensitivity of the task response to that error.} To separate these two effects, we decompose $\Delta \boldsymbol{\Sigma}_i$ into its total magnitude and normalized directional profile. First, define $q_i \triangleq {\rm tr}(\Delta \boldsymbol{\Sigma}_i)$. Thus, $q_i$ measures the total reduction in geometry mean-squared error achieved by refining building $i$. We refer to $q_i$ as its refinement quality; a larger $q_i$ indicates that the prescribed refinement operation can correct more geometry error. For $q_i > 0$, define $\bm{{\rm U}}_i \triangleq \Delta \boldsymbol{\Sigma}_i / q_i$. Thus, we know that ${\rm tr}(\bm{{\rm U}}_i) = 1$. This normalization assigns the total error reduction magnitude to $q_i$, while $\bm{{\rm U}}_i$ describes how that reduction is distributed across the geometric parameter directions. Therefore, $v_i$ can be rewritten as
\begin{equation}\label{eq:decomposition3}
	\setlength\abovedisplayskip{3pt}
	\setlength\belowdisplayskip{3pt}
	v_i = q_i \mathbb{E}_{\eta}\left[{\rm tr}\left(\bm{{\rm J}}_i^{\mathsf{H}}(\eta) \bm{{\rm W}} \bm{{\rm J}}_i(\eta) \bm{{\rm U}}_i\right)\right] = q_is_i,
\end{equation}
where $s_i \triangleq \mathbb{E}_{\eta}[{\rm tr}(\bm{{\rm J}}_i^{\mathsf{H}}(\eta) \bm{{\rm W}} \bm{{\rm J}}_i(\eta) \bm{{\rm U}}_i)]$. From (\ref{eq:decomposition3}), we know that \textit{a building with substantial correctable geometry error may still have little refinement value if the wireless task is insensitive to that geometry error. Conversely, a task-sensitive building may also offer little improvement if its initial geometry is already accurate.}

Since $\widetilde{R}(\bm{{\rm 0}})$ in (\ref{eq:rz}) is independent of the decision vector $\bm{{\rm z}}$, minimizing the surrogate risk in (\ref{eq:rz}) is equivalent to maximizing the total geometry-induced risk reduction\footnote{Although the residual non-geometric risk $R_b$ affects the absolute WDT fidelity, it does not affect the refinement decision under the stated orthogonality assumptions.}. This gives
\begin{subequations}
	\setlength\abovedisplayskip{3pt}
	\setlength\belowdisplayskip{3pt}
	\begin{align}
		(\textbf{P2}): \max_{\bm{{\rm z}}} \quad & \sum_{i\in \mathcal{I}}z_iq_is_i \\
		\mathrm{s.t.} \quad & (\ref{constraint1}), (\ref{constraint2})
	\end{align}
\end{subequations}
Problem (\textbf{P2}) exposes a clear structure for the refinement decision, however, the coefficients in the objective function are still unknown before refinement. In particular, $q_i$ and $\bm{{\rm U}}_i$ depend on the attainable geometry improvement, while $\bm{{\rm J}}_i(\eta)$ characterizes the local response sensitivity.

\subsection{Practical Refinement-Value Estimation}
Since the building-specific second-moment reduction $\Delta\boldsymbol{\Sigma}_i$ is unavailable before refinement, directly evaluating $v_i$ remains impractical. In this work, we approximate $v_i$ using three physically interpretable building-level transformation directions. Specifically, we consider the transformation set $\mathcal{A}=\{x, y, h\}$, corresponding to horizontal translations along the $x$- and $y$-axes and a variation in building height, respectively. These transformations are intended to probe physically meaningful geometry variations that can affect the propagation response.

Let $\Psi_{i,a,t}(\cdot)$ denote the operator that transforms building $i$ according to mode $a\in \mathcal{A}$, where $t$ is the signed transformation amplitude. Applying this transformation to the initial geometry gives
\begin{equation}
	\setlength\abovedisplayskip{3pt}
	\setlength\belowdisplayskip{3pt} 
	\hat{\mathcal{G}}_{i,a,t} \triangleq \Psi_{i,a,t}\left(\hat{\mathcal{G}}_i^{(0)}\right).
\end{equation}
Let $\mathcal{D}_{i,a,t}$ denote the WDT obtained by replacing $\hat{\mathcal{G}}_i^{(0)}$ in $\mathcal{D}_0$ with $\hat{\mathcal{G}}_{i,a,t}$, while keeping all other WDT components unchanged. Its task response is denoted as 
\begin{equation}
	\setlength\abovedisplayskip{3pt}
	\setlength\belowdisplayskip{3pt} 
	\hat{\bm{{\rm y}}}_{i,a,t}(\eta) \triangleq \mathcal{F}_{\mathcal{T}}(\mathcal{D}_{i,a,t}, \eta).
\end{equation}
The corresponding direction of change in the geometry-parameter vector under transformation $a$ is defined as
\begin{equation}
	\setlength\abovedisplayskip{3pt}
	\setlength\belowdisplayskip{3pt} 
	\bm{{\rm b}}_{i,a} \triangleq \left.\frac{\partial}{\partial t} \phi_i\left(\hat{\mathcal{G}}_{i,a,t}\right) \right|_{t=0}
\end{equation}
We model the unknown reduction in the second moment of geometry error by the following structured low-rank approximation:
\begin{equation}\label{eq:low_rank_approximation}
	\setlength\abovedisplayskip{3pt}
	\setlength\belowdisplayskip{3pt} 
	\Delta\boldsymbol{\Sigma}_i \approx \frac{\sigma_g^2}{|\mathcal{A}|}\sum_{a\in\mathcal{A}} \bm{{\rm b}}_{i,a} \bm{{\rm b}}_{i,a}^{\mathsf{H}},
\end{equation}
where $\sigma_g^2>0$ represents the common second-moment reduction associated with the standardized transformation amplitudes. This model assumes that the correctable component of the geometry error is primarily distributed along the selected building-level transformation directions. Substituting this structured model into the expression of $v_i$ gives
\begin{equation}
	\setlength\abovedisplayskip{3pt}
	\setlength\belowdisplayskip{3pt} 
	\begin{aligned}
		v_i &\approx \frac{\sigma_g^2}{|\mathcal{A}|}\sum_{a\in\mathcal{A}}\mathbb{E}_{\eta}\left[{\rm tr}\left(\bm{{\rm J}}_i^{\mathsf{H}}(\eta) \bm{{\rm W}} \bm{{\rm J}}_i(\eta)\bm{{\rm b}}_{i,a} \bm{{\rm b}}_{i,a}^{\mathsf{H}} \right)\right] \\
		&=\frac{\sigma_g^2}{|\mathcal{A}|}\sum_{a\in\mathcal{A}}\mathbb{E}_{\eta}\left[\Vert \bm{{\rm J}}_i(\eta)\bm{{\rm b}}_{i,a} \Vert^2_{\bm{{\rm W}}}\right] = \overline{v}_i.
	\end{aligned}
\end{equation}
Therefore, under the structured directional model, the refinement value of building $i$ is determined by the task-response sensitivities along the selected transformation directions.

Based on the chain rule, the response derivative along the transformation $a$ satisfies
\begin{equation}
	\setlength\abovedisplayskip{3pt}
	\setlength\belowdisplayskip{3pt} 
	\left.\frac{\partial \hat{\bm{{\rm y}}}_{i,a,t}(\eta)}{\partial t} \right|_{t=0} = \bm{{\rm J}}_i(\eta) \bm{{\rm b}}_{i,a}.
\end{equation}
This directional derivative can be estimated without explicitly constructing either $\bm{{\rm J}}_i(\eta)$ or $\bm{{\rm b}}_{i,a}$. Using paired perturbation with step size $\delta_a>0$, we define
\begin{equation}\label{eq:xyz_central_difference}
	\setlength\abovedisplayskip{3pt}
	\setlength\belowdisplayskip{3pt} 
	\hat{\bm{{\rm d}}}_{i,a}(\eta) \triangleq \frac{\hat{\bm{{\rm y}}}_{i,a,+\delta_a}(\eta) - \hat{\bm{{\rm y}}}_{i,a,-\delta_a}(\eta) }{2\delta_a},
\end{equation}
Provided that the task-response operator is locally smooth along the considered transformation, the central-difference estimate satisfies
\begin{equation}
	\setlength\abovedisplayskip{3pt}
	\setlength\belowdisplayskip{3pt} 
	\hat{\bm{{\rm d}}}_{i,a}(\eta) = \bm{{\rm J}}_i(\eta) \bm{{\rm b}}_{i,a} + O(\delta_a^2).
\end{equation}
Accordingly, we define the practical structured sensitivity score as
\begin{equation}
	\setlength\abovedisplayskip{3pt}
	\setlength\belowdisplayskip{3pt} 
	\hat{s}_i \triangleq \frac{1}{|\mathcal{A}|}\sum_{a\in\mathcal{A}} \mathbb{E}_{\eta}\left[\Vert \hat{\bm{{\rm d}}}_{i,a}(\eta)  \Vert^2_{\bm{{\rm W}}}\right].
\end{equation}
Combining the structured second-moment model with the central-difference approximation, we define the practical refinement-value estimate as $\hat v_i\triangleq\sigma_g^2\hat s_i$. Since the common factor $\sigma_g^2$ does not affect the ordering of candidate buildings, minimizing the surrogate risk is approximately equivalent to solving
\begin{equation}
	\setlength\abovedisplayskip{3pt}
	\setlength\belowdisplayskip{3pt} 
	\begin{aligned}
		(\textbf{P3}): \max_{\bm{{\rm z}}} \quad & \sum_{i\in\mathcal{I}}z_i \hat{s}_i \quad \mathrm{s.t.} \space  (\ref{constraint1}), (\ref{constraint2})
	\end{aligned}
\end{equation}
Thus, $\hat{s}_i$ should be interpreted as a structured directional surrogate for the expected task-risk reduction, rather than as an estimate of the complete geometry-response Jacobian.

\subsection{Propagation-Relevance Ellipsoid-Guided Screening}
Although the central-difference estimator requires only six response evaluations per building for each wireless system configuration, applying it to every building is time-consuming, especially in large urban scenes. Therefore, we introduce a geometry-based screening stage before the perturbation evaluations. This stage identifies candidate buildings according to their potential involvement in Tx–Rx propagation, while their final refinement priorities are determined by the central-difference scores $\hat{s}_i$.

\begin{definition}\label{Definition1}
	For a Tx at $\mathbf{x}_T\in \mathbb{R}^{3\times 1}$ and an Rx at $\mathbf{x}_R\in \mathbb{R}^{3\times 1}$, the propagation-relevance ellipsoid is defined as
	\begin{equation}
		\setlength\abovedisplayskip{3pt}
		\setlength\belowdisplayskip{3pt}
		\Omega_\Delta(\mathbf{x}_T, \mathbf{x}_R)\!\triangleq\!\{\mathbf{x}\in\mathbb{R}^3 \!\!: \!\! \Vert \mathbf{x} - \mathbf{x}_T \Vert + \Vert \mathbf{x} - \mathbf{x}_R \Vert \! \leq \! \Vert \mathbf{x}_T \!-\! \mathbf{x}_R \Vert + \Delta \},
	\end{equation}
	where $\Delta > 0$ specifies the maximum excess path length relative to the direct Tx–Rx distance $\|\mathbf{x}_T-\mathbf{x}_R\|$. 
\end{definition}

Let $\mathcal{P}=(\mathbf{x}_T, \mathbf{p}_1, \mathbf{p}_2, \cdots, \mathbf{p}_K, \mathbf{x}_R)$ denote an arbitrary finite-order propagation path from $\mathbf{x}_T$ to $\mathbf{x}_R$. Here, $\mathbf{p}_k\in \mathbb{R}^{3\times 1}$ denotes the coordinate of the $k$-th interaction point, which may correspond to reflection, diffraction, or scattering, and $K$ is the number of interaction points. The path length can be given as
\begin{equation}
	\setlength\abovedisplayskip{3pt}
	\setlength\belowdisplayskip{3pt}
	L(\mathcal{P}) = \Vert \mathbf{p}_1 - \mathbf{x}_T \Vert + \sum_{k=1}^{K-1} \Vert \mathbf{p}_{k+1} - \mathbf{p}_{k} \Vert + \Vert \mathbf{x}_R - \mathbf{p}_{K}\Vert.
\end{equation}
For $K=0$, the path reduces to the direct Tx-Rx segment with $L(\mathcal{P}) = \Vert \mathbf{x}_R - \mathbf{x}_T \Vert$. Thus, given the definition of the propagation path, we introduce the following lemma.

\begin{lemma}\label{Lemma1}
	If a propagation path $\mathcal{P}$ satisfies $L(\mathcal{P}) \leq \Vert \mathbf{x}_T - \mathbf{x}_R \Vert + \Delta$, then every interaction point $\mathbf{p}_k, k=1,2,\cdots,K$ of the path, if any, lies within the propagation-relevance ellipsoid $\Omega_\Delta(\mathbf{x}_T, \mathbf{x}_R)$ defined in \textbf{Definition} \ref{Definition1}.
\end{lemma}

\begin{IEEEproof}
	Consider an arbitrary interaction point $\mathbf{p}_k, k=1,2,\cdots,K$ on the path. The path length from $\mathbf{x}_T$ to $\mathbf{p}_{k}$ is
	\begin{equation}
		\setlength\abovedisplayskip{3pt}
		\setlength\belowdisplayskip{3pt}
		L(\mathcal{P}_{T,k}) = \Vert \mathbf{p}_1 - \mathbf{x}_T \Vert + \sum_{j=1}^{k-1} \Vert \mathbf{p}_{j+1} - \mathbf{p}_{j} \Vert.
	\end{equation}
	The remaining path length from $\mathbf{p}_{k}$ to $\mathbf{x}_R$ is
	\begin{equation}
		\setlength\abovedisplayskip{3pt}
		\setlength\belowdisplayskip{3pt}
		L(\mathcal{P}_{k,R}) = \sum_{j=k}^{K-1} \Vert \mathbf{p}_{j+1} - \mathbf{p}_{j} \Vert + \Vert \mathbf{x}_R - \mathbf{p}_{K} \Vert. 
	\end{equation}
	By construction, we have
	\begin{equation}\label{eq:total_path_length}
		\setlength\abovedisplayskip{3pt}
		\setlength\belowdisplayskip{3pt}
		L(\mathcal{P}) = L(\mathcal{P}_{T,k}) + L(\mathcal{P}_{k,R}).
	\end{equation}
	Since $L(\mathcal{P}_{T,k})$ is the length of an actual path connecting $\mathbf{x}_T$ and $\mathbf{p}_k$, based on the triangle inequality, the Euclidean distance between $\mathbf{x}_T$ and $\mathbf{p}_k$ cannot exceed this path length. Therefore, we have
	\begin{equation}
		\setlength\abovedisplayskip{3pt}
		\setlength\belowdisplayskip{3pt}
		\Vert \mathbf{p}_{k} - \mathbf{x}_T \Vert \leq L(\mathcal{P}_{T,k}).
	\end{equation}
	Similarly, since $L(\mathcal{P}_{k,R})$ is the length of an actual path connecting $\mathbf{p}_k$ and $\mathbf{x}_R$, we have
	\begin{equation}
		\setlength\abovedisplayskip{3pt}
		\setlength\belowdisplayskip{3pt}
		\Vert \mathbf{p}_{k} - \mathbf{x}_R \Vert \leq L(\mathcal{P}_{k,R}).
	\end{equation}
	Combining these inequalities with (\ref{eq:total_path_length}) yields
	\begin{equation}
		\setlength\abovedisplayskip{3pt}
		\setlength\belowdisplayskip{3pt}
		\Vert \mathbf{p}_{k} - \mathbf{x}_T \Vert + \Vert \mathbf{p}_{k} - \mathbf{x}_R \Vert \leq L(\mathcal{P}_{T,k}) + L(\mathcal{P}_{k,R}) = L(\mathcal{P}).
	\end{equation}
	Given that the path length satisfies $L(\mathcal{P}) \leq \Vert \mathbf{x}_T - \mathbf{x}_R \Vert + \Delta$, it follows that
	\begin{equation}
		\setlength\abovedisplayskip{3pt}
		\setlength\belowdisplayskip{3pt}
		\Vert \mathbf{p}_{k} - \mathbf{x}_T \Vert + \Vert \mathbf{p}_{k} - \mathbf{x}_R \Vert \leq \Vert \mathbf{x}_T - \mathbf{x}_R \Vert + \Delta.
	\end{equation}
	According to \textbf{Definition} \ref{Definition1}, we have $\mathbf{p}_{k}\in \Omega_\Delta(\mathbf{x}_T, \mathbf{x}_R)$. Since $\mathbf{p}_{k}$ was chosen arbitrarily, the result holds for every interaction point on the path. The proof is thus completed.
\end{IEEEproof}

Consequently, a building supporting an interaction on such a path must intersect the ellipsoid. The ellipsoid also contains the direct Tx-Rx segment, providing a unified geometric region for screening potential NLoS multipath interactions and LoS blockage. These properties motivate ellipsoid overlap as a geometric screening criterion. It should be noted that the intersection condition is necessary for the considered propagation interactions, but does not by itself establish a valid propagation path.

\subsubsection{Polar-domain Rx Proxy Set} 
Next, we construct a representative Rx proxy set from the initial low-fidelity WDT to evaluate the screening criterion efficiently. The following construction considers a fixed Tx location $\bm{{\rm x}}_T$.

Let $\mathcal{V}=\{V_j\}_{j=1}^{J}$ denote $J$ triangular cells of the terrain measurement surface~\cite{Aoudia2025sionna}. A candidate Rx location is generated from each cell as
\begin{equation}
	\setlength\abovedisplayskip{3pt}
	\setlength\belowdisplayskip{3pt}
	\mathbf{c}_j  = \mathbf{t}_j + h_r \mathbf{e}_z, j=1,2,\cdots,J,
\end{equation} 
where $\mathbf{t}_j\in \mathbb{R}^{3\times 1}$ is the cell center, $h_r>0$ is the Rx height, and $\mathbf{e}_z = [0, 0, 1]^{\mathsf{T}}$. Candidates whose horizontal coordinates lie inside the building footprints are discarded. For each remaining candidate, its Tx-center polar coordinates are 
\begin{equation}
	\setlength\abovedisplayskip{3pt}
	\setlength\belowdisplayskip{3pt}
	\left\{
	\begin{array}{lr}
		r_j = \sqrt{(c_j^{(x)} - x_T^{(x)})^2 + (c_j^{(y)} - x_T^{(y)})^2}, &  \\
		\alpha_j = \operatorname{atan2}(c_j^{(y)} - x_T^{(y)}, c_j^{(x)} - x_T^{(x)}), &  
	\end{array}
	\right.
\end{equation}
where $x_T^{(x)}$ and $x_T^{(y)}$ represent the coordinates along the $x$- and $y$-axes of Tx located at $\mathbf{x}_T$. The candidates are first partitioned into $L$ radial \emph{bands} using boundaries $0=R_1 < R_2 < \cdots < R_L$. Thus, the radial bands are
\begin{equation}
	\setlength\abovedisplayskip{3pt}
	\setlength\belowdisplayskip{3pt}
	\left\{
	\begin{array}{lr}
		\mathcal{R}_l=\{\mathbf{c}_j \vert R_l \leq r_j < R_{l+1}\}, l=1,\cdots,L-1, &  \\
		\mathcal{R}_L=\{\mathbf{c}_j \vert r_j \geq R_L\}. &  
	\end{array}
	\right.
\end{equation}
Each radial band is further assigned a target spacing $\Delta_{r_l} > 0$. Finer spacing is used near the transmitter and coarser spacing farther away. Within each band, the candidates are divided into radial \emph{rings} and angular \emph{sectors}, forming \emph{polar cells}. One representative is retained from each nonempty cell, selected as the candidate closest to its corresponding cell center. Collecting the representatives gives the Rx proxy set $\mathcal{X}_R^p \triangleq \{\bm{{\rm x}}_{R, n}^p\}_{n=1}^N$, where $N$ is the number of retained locations. This set is used for geometric screening.

\subsubsection{Ellipsoid-Overlap Score}
To calculate the overlap score, we need to construct a lightweight geometric proxy for each building. This is motivated by the fact that meshes in a low-fidelity WDT may contain reconstruction artifacts, such as holes, fragmented facets, and non-watertight surfaces. Directly performing mesh-level geometric analysis on such imperfect meshes can be unstable and computationally expensive. To this end, for each building $i\in \mathcal{I}$, we construct a lightweight occupancy proxy from its initial digital geometry $\hat{\mathcal{G}}_i^{(0)}$:
\begin{equation}
	\label{eq:building_proxy}
	\setlength\abovedisplayskip{3pt}
	\setlength\belowdisplayskip{3pt}
	\widetilde{B}_i = \mathcal{F}_i\times[h_i^{-},h_i^{+}],
\end{equation}
where $\mathcal{F}_i\subseteq \mathbb{R}^2$ is the convex-hull footprint of the building geometry, and $h_i^-$ and $h_i^+$ are its lower and upper elevations. This proxy captures the dominant building extent while reducing sensitivity to holes and fragmented mesh surfaces.

For the $n$-th Tx-Rx proxy pair, define the normalized ellipsoid-overlap score as
\begin{equation}\label{eq:norm_overlap_vol}
	\setlength\abovedisplayskip{3pt}
	\setlength\belowdisplayskip{3pt}
	\rho_{i,n}^{\mathrm{E}} \triangleq \frac{{\mathrm{Vol}}(\widetilde{B}_i \cap \Omega_\Delta(\mathbf{x}_T, \mathbf{x}_{R, n}^p))}{{\mathrm{Vol}}(\Omega_\Delta(\mathbf{x}_T, \mathbf{x}_{R,n}^p))},
\end{equation}
where ${\rm Vol}(\cdot)$ denotes the 3D volume. The intersection volume can be evaluated over the building footprint:
\begin{equation}\label{eq:overlap_vol}
	\setlength\abovedisplayskip{3pt}
	\setlength\belowdisplayskip{3pt}
	{\mathrm{Vol}}(\widetilde{B}_i \cap \Omega_\Delta(\mathbf{x}_T, \mathbf{x}_{R,n}^p)) = \int_{\mathcal{F}_i}\tau_{i,n}(\mathbf{v})d\mathbf{v},
\end{equation}
where $\bm{{\rm v}}\in\mathbb{R}^2$ is a horizontal location and $\tau_{i,n}(\mathbf{v})$ is the vertical overlap thickness between the building proxy and the ellipsoid. Using $Q_i$ uniformly distributed footprint samples $\{\bm{{\rm v}}_{i,q}\}_{q=1}^{Q_i}$, the volume is then approximated by
\begin{equation}\label{eq:mc_overlap_vol}
	\setlength\abovedisplayskip{3pt}
	\setlength\belowdisplayskip{3pt}
	{\mathrm{Vol}}(\widetilde{B}_i \cap \Omega_\Delta(\mathbf{x}_T, \mathbf{x}_{R,n}^p)) \approx \frac{{\rm Area}(\mathcal{F}_i)}{Q_i}\sum_{q=1}^{Q_i} \tau_{i,n}(\mathbf{v}_{i,q}),
\end{equation}
where ${\rm Area}(\mathcal{F}_i)$ denotes the footprint area. Although the ellipsoid contains the LoS component, the overlap volume does not explicitly distinguish buildings that block the direct LoS link. We therefore retain an additional LoS weight,
\begin{equation}
	\setlength\abovedisplayskip{3pt}
	\setlength\belowdisplayskip{3pt}
	\label{eq:los_boosting_factor}
	\chi_{i,n}^{\mathrm{LoS}}
	\!=\!\!
	\begin{cases}
		\eta_{\mathrm{LoS}}, \!& \text{if $\widetilde{B}_i$ is the primary LoS blocker},\\
		1, \! & \text{otherwise},
	\end{cases}
\end{equation}
where $\eta_{\mathrm{LoS}} > 1$ is the LoS boosting factor. The primary blocker is the first building proxy intersected along the segment from $\bm{{\rm x}}_T$ to $\bm{{\rm x}}_{R,n}^p$. The screening score is then given by
\begin{equation}
	\setlength\abovedisplayskip{3pt}
	\setlength\belowdisplayskip{3pt}
	K_i^E \triangleq \frac{1}{N}\sum_{n=1}^N \chi_{i,n}^{\mathrm{LoS}}\rho_{i,n}^E.
\end{equation}
The overlap term measures geometric participation in the propagation-relevant region, while the LoS weight additionally emphasizes direct-link blockers. This score is computed entirely from the initial WDT geometry, without perturbation-response evaluations. Let $M_c$ denote the prescribed size of the screened building candidate set, where $1\leq M_c \leq N_b$. We retain the $M_c$ buildings with the highest screening scores as $\mathcal{I}_c \triangleq {\rm Top}_{M_c}\left(\{K_i^E\}_{i\in \mathcal{I}}\right)$, where ${\rm Top}_{M_c}$ returns the corresponding building indices.

Restricting (\textbf{P3}) to the screened candidate set $\mathcal{I}_c$ yields the following problem:
\begin{subequations}
	\setlength\abovedisplayskip{3pt}
	\setlength\belowdisplayskip{3pt} 
	\begin{align}
		(\textbf{P4}): \max_{\{z_i\}_{i\in \mathcal{I}_c}} \quad & \sum_{i\in\mathcal{I}_c}z_i \hat{s}_i \\
		\mathrm{s.t.} \quad & z_i \in \{0, 1\}, i\in \mathcal{I}_c, \label{constraint4_1}\\
		& \sum_{i\in\mathcal{I}_c}c_iz_i \leq B.  \label{constraint4_2}
	\end{align}
\end{subequations}
In (\textbf{P4}), $K_i^E$ determines which buildings should undergo sensitivity evaluation, whereas $\hat{s}_i$ determines which candidates should be refined. Overall, the proposed method transforms the original task-oriented refinement problem in (\textbf{P1}) into a tractable selection surrogate problem in (\textbf{P4}). The pipeline of the proposed TONR framework is illustrated in Fig. \ref{fig:pipeline}.

\begin{remark}
	When all candidate buildings have the same refinement cost, i.e., $c_i = c$, the budget constraint in (\textbf{P4}) reduces to $\sum_{i\in\mathcal{I}_c} z_i \leq \lfloor B/c \rfloor$. Therefore, (\textbf{P4}) can be solved by selecting the buildings with the largest predicted sensitivity scores. In practice, the refinement costs may vary with the building size, sensing difficulty, and the computation, storage, or acquisition resources required for reconstruction. These costs can be estimated before refinement using information available from the initial WDT. For example, when refinement increases the point-cloud density from $\rho_0$ to $\rho_1$, a simple cost estimate is $\hat{c}_i = \kappa(\rho_1-\rho_0)\hat{A}^{\mathrm{ext}}_i$, where $\hat{A}^{\mathrm{ext}}_i$ is the proxy exterior area estimated from
	the initial WDT. For a building represented by a footprint with area $\hat{A}^{\mathrm{fp}}_i$, perimeter $\hat{P}_i$, and height
	$\hat{h}_i$, it can be approximated as $\hat{A}^{\mathrm{ext}}_i = \hat{A}^{\mathrm{fp}}_i+\hat{P}_i\hat{h}_i$, where the inaccessible bottom surface is excluded. With these heterogeneous costs, (\textbf{P4}) becomes a standard 0--1 knapsack problem and can be solved by dynamic programming for discretized costs or by a mixed-integer linear program for general costs.
\end{remark}

\begin{figure*}[t]
	\centering
	\includegraphics[width=7.0in]{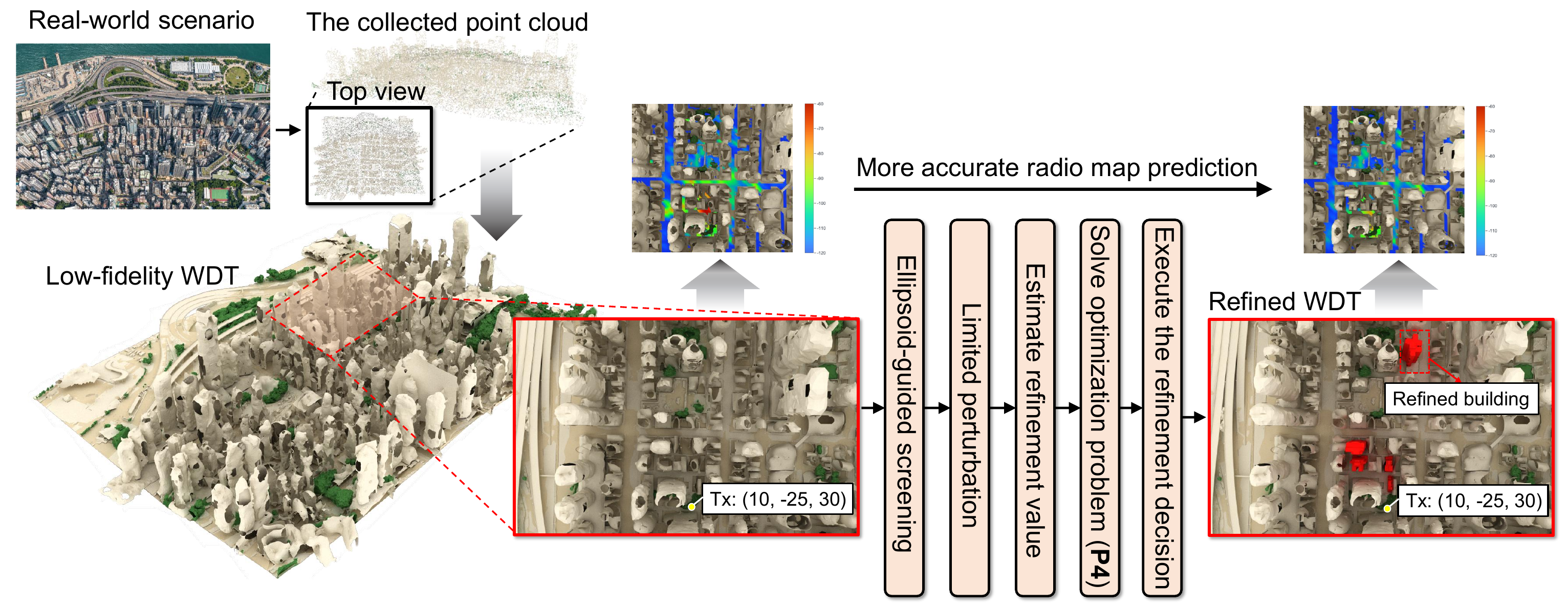}
	\caption{Pipeline of the proposed method.}
	\label{fig:pipeline}
\end{figure*}

\subsection{Optimality Gap Analysis}
This subsection quantifies the suboptimality of the refinement decision obtained from (\textbf{P4}) relative to the original task-risk expectation minimization problem (\textbf{P1}). The analysis accounts for four sources of approximation error: response linearization, structured second-moment modeling, central-difference estimation, and candidate screening.

\subsubsection{First-order risk approximation} Recall that the expected task risk is
\begin{equation}\label{eq:true_task-risk}
	\setlength\abovedisplayskip{3pt}
	\setlength\belowdisplayskip{3pt} 
	R(\bm{{\rm z}}) \triangleq \mathbb{E}_{\eta, \bm{{\rm e}}}\left[\Vert \hat{\bm{{\rm y}}}(\boldsymbol{\gamma}(\bm{{\rm z}}); \eta) - \bm{{\rm y}}^\star(\eta) \Vert^2_{\bm{{\rm W}}}\right].
\end{equation}
Using the exact response decomposition in (\ref{eq:exact_response_decoposition}), we obtain
\begin{equation}
	\setlength\abovedisplayskip{3pt}
	\setlength\belowdisplayskip{3pt} 
	R(\bm{{\rm z}}) = \mathbb{E}_{\eta, \bm{{\rm e}}}\left[\Big\Vert \sum_{i\in \mathcal{I}}\bm{{\rm J}}_i(\eta)\bm{{\rm e}}_i(z_i) + \bm{{\rm r}}_{\bm{{\rm z}}}(\eta) - \bm{{\rm r}}^\star(\eta) + \bm{{\rm b}}(\eta) \Big\Vert^2_{\bm{{\rm W}}}\right].
\end{equation}
Expanding the weighted squared norm and recalling the definition of $\widetilde{R}(\bm{{\rm z}})$ in (\ref{eq:surrogate1}) yield 
\begin{equation}
	\setlength\abovedisplayskip{3pt}
	\setlength\belowdisplayskip{3pt} 
	\begin{aligned}
		& R(\bm{{\rm z}}) - \widetilde{R}(\bm{{\rm z}}) \\
		& = 2\Re\left\{\!\mathbb{E}_{\eta, \bm{{\rm e}}}\!\!\left[\left(\sum_{i\in \mathcal{I}}\bm{{\rm J}}_i(\eta)\bm{{\rm e}}_i(z_i) \!+ \!\!\bm{{\rm b}}(\eta) \right)^{\mathsf{H}}\!\!\bm{{\rm W}}\!(\bm{{\rm r}}_{\bm{{\rm z}}}(\eta) \!-\! \bm{{\rm r}}^\star(\eta))\right]\right\} \\
		&+\mathbb{E}_{\eta, \bm{{\rm e}}}\left[\Vert \bm{{\rm r}}_{\bm{{\rm z}}}(\eta) - \bm{{\rm r}}^\star(\eta) \Vert^2_{\bm{{\rm W}}}\right].
	\end{aligned}
\end{equation}
By the Cauchy–Schwarz inequality, we have
\begin{equation}
	\setlength\abovedisplayskip{3pt}
	\setlength\belowdisplayskip{3pt} 
	\begin{aligned}
	\vert R(\bm{{\rm z}}) - \widetilde{R}(\bm{{\rm z}}) \vert \leq & 2\sqrt{\widetilde{R}(\bm{{\rm z}})}\left(\mathbb{E}_{\eta, \bm{{\rm e}}}\left[\Vert \bm{{\rm r}}_{\bm{{\rm z}}}(\eta) - \bm{{\rm r}}^\star(\eta) \Vert^2_{\bm{{\rm W}}}\right]\right)^{1/2} \\
	& + \mathbb{E}_{\eta, \bm{{\rm e}}}\left[\Vert \bm{{\rm r}}_{\bm{{\rm z}}}(\eta) - \bm{{\rm r}}^\star(\eta) \Vert^2_{\bm{{\rm W}}}\right].
	\end{aligned}
\end{equation}
Assuming these moments are finite, define the uniform risk-approximation bound over the feasible action set $\mathcal{Z}$ as $\epsilon_{{\rm lin}}$. Consequently, we have
\begin{equation}\label{eq:epsilon_lin}
	\setlength\abovedisplayskip{3pt}
	\setlength\belowdisplayskip{3pt} 
	\vert R(\bm{{\rm z}}) - \widetilde{R}(\bm{{\rm z}}) \vert \leq \epsilon_{{\rm lin}}, \forall \bm{{\rm z}}\in \mathcal{Z}.
\end{equation}
This bound accounts for both the squared Taylor remainder difference and its interaction with the first-order response error, including the residual non-geometric mismatch.

\subsubsection{Structured second-moment approximation} 
First, define the three-direction model in (\ref{eq:low_rank_approximation}) as $\overline{\Delta \boldsymbol{\Sigma}_i} \triangleq \frac{\sigma_g^2}{|\mathcal{A}|}\sum_{a\in\mathcal{A}} \bm{{\rm b}}_{i,a} \bm{{\rm b}}_{i,a}^{\mathsf{H}}$, and let $\bm{{\rm H}}_i \triangleq \mathbb{E}_{\eta}[\bm{{\rm J}}^{\mathsf{H}}_i(\eta)\bm{{\rm W}}\bm{{\rm J}}_i(\eta) ]$. The refinement value associated with exact directional derivatives is $\overline{v_i} \triangleq {\rm tr}(\bm{{\rm H}}_i \overline{\Delta \boldsymbol{\Sigma}_i})$. Let $\bm{{\rm E}}_i \triangleq \Delta \boldsymbol{\Sigma}_i - \overline{\Delta \boldsymbol{\Sigma}_i}$. The corresponding error in the refinement value can be expressed as
\begin{equation}\label{eq:structured_error_v}
	\setlength\abovedisplayskip{3pt}
	\setlength\belowdisplayskip{3pt} 
	\begin{aligned}
	\vert v_i - \overline{v_i} \vert &= \vert {\rm tr}(\bm{{\rm H}}_i \Delta \boldsymbol{\Sigma}_i) -  {\rm tr}(\bm{{\rm H}}_i \overline{\Delta \boldsymbol{\Sigma}_i}) \vert \\
	& = \vert {\rm tr}(\bm{{\rm H}}_i \bm{{\rm E}}_i) \vert \leq \Vert \bm{{\rm H}}_i \Vert_{{\rm F}} \Vert \bm{{\rm E}}_i \Vert_{{\rm F}} \triangleq \epsilon_i^{{\rm str}}.
	\end{aligned}
\end{equation}
To account for the refinement budget, define the accumulated structured-error bound as
\begin{equation}
	\setlength\abovedisplayskip{3pt}
	\setlength\belowdisplayskip{3pt} 
	\epsilon_{{\rm str}} \triangleq \max_{\bm{{\rm z}} \in \mathcal{Z}_c}\sum_{i\in \mathcal{I}_c} z_i\epsilon_i^{{\rm str}},
\end{equation}
where $\mathcal{Z}_c = \{\bm{{\rm z}}\in \mathcal{Z}: z_i=0, \forall i\notin \mathcal{I}_c\}$. This error captures discrepancies between the actual second-moment reduction and its three-direction approximation. These discrepancies include omitted geometric directions and inaccurate reduction magnitudes within the selected directions.

\subsubsection{Central-difference approximation}
Denote the directional derivative estimation error under the adopted perturbation direction $a$ as
\begin{equation}
	\setlength\abovedisplayskip{3pt}
	\setlength\belowdisplayskip{3pt} 
	\bm{{\rm f}}_{i,a}(\eta) \triangleq \hat{\bm{{\rm d}}}_{i,a}(\eta) - \bm{{\rm J}}_i(\eta)\bm{{\rm b}}_{i,a}.
\end{equation}
For one direction, we have
\begin{equation}
	\setlength\abovedisplayskip{3pt}
	\setlength\belowdisplayskip{3pt} 
	\begin{aligned}
	\Vert \hat{\bm{{\rm d}}}_{i,a}(\eta) \Vert^2_{\bm{{\rm W}}} &= \Vert \bm{{\rm J}}_i(\eta)\bm{{\rm b}}_{i,a} \Vert^2_{\bm{{\rm W}}} \\
	&+2\Re\left\{\left(\bm{{\rm J}}_i(\eta)\bm{{\rm b}}_{i,a}\right)^{\mathsf{H}}\bm{{\rm W}} \bm{{\rm f}}_{i,a}(\eta) \right\} \\
	&+\Vert \bm{{\rm f}}_{i,a}(\eta) \Vert^2_{\bm{{\rm W}}}.
	\end{aligned}
\end{equation}
For each $a$, we have
\begin{equation}
	\setlength\abovedisplayskip{3pt}
	\setlength\belowdisplayskip{3pt} 
	\begin{aligned}
	&\mathbb{E}_{\eta}\left[\Vert \hat{\bm{{\rm d}}}_{i,a}(\eta) \Vert^2_{\bm{{\rm W}}} - \Vert \bm{{\rm J}}_i(\eta) \bm{{\rm b}}_{i,a} \Vert^2_{\bm{{\rm W}}} \right]\\
	&=2\Re\left\{\mathbb{E}_{\eta}\left[\left(\bm{{\rm J}}_i(\eta)\bm{{\rm b}}_{i,a}\right)^{\mathsf{H}}\bm{{\rm W}} \bm{{\rm f}}_{i,a}(\eta)\right]\right\} + F_{i,a}^2,
\end{aligned}
\end{equation}
where we define $\mathbb{E}_{\eta}[\Vert \bm{{\rm f}}_{i,a}(\eta) \Vert^2_{\bm{{\rm W}}}]$ as $F_{i,a}^2$. Similarly, by the Cauchy--Schwarz inequality, we obtain
\begin{equation}
	\setlength\abovedisplayskip{3pt}
	\setlength\belowdisplayskip{3pt} 
	\begin{aligned}
		&\Big\vert \mathbb{E}_{\eta}\left[\left(\bm{{\rm J}}_i(\eta)\bm{{\rm b}}_{i,a}\right)^{\mathsf{H}}\bm{{\rm W}} \bm{{\rm f}}_{i,a}(\eta)\right] \Big\vert \\
		&\leq \mathbb{E}_{\eta}\left[\Vert \bm{{\rm J}}_i(\eta)\bm{{\rm b}}_{i,a}\Vert_{\bm{{\rm W}}} \Vert \bm{{\rm f}}_{i,a}(\eta) \Vert_{\bm{{\rm W}}}   \right] \\
		&\leq \sqrt{\mathbb{E}_{\eta}\left[\Vert \bm{{\rm J}}_i(\eta)\bm{{\rm b}}_{i,a} \Vert^2_{\bm{{\rm W}}}\right]} \sqrt{\mathbb{E}_{\eta}\left[\Vert \bm{{\rm f}}_{i,a}(\eta) \Vert^2_{\bm{{\rm W}}}\right]} = D_{i,a}F_{i,a},
	\end{aligned}
\end{equation}
where $D_{i,a}\triangleq \sqrt{\mathbb{E}_{\eta}\left[\Vert \bm{{\rm J}}_i(\eta)\bm{{\rm b}}_{i,a} \Vert^2_{\bm{{\rm W}}}\right]}$. Based on $|\Re\{{x}\}| \leq |x|$, we have
\begin{equation}
	\setlength\abovedisplayskip{3pt}
	\setlength\belowdisplayskip{3pt} 
	\Big\vert \mathbb{E}_{\eta}\left[\Vert \hat{\bm{{\rm d}}}_{i,a}(\eta) \Vert^2_{\bm{{\rm W}}} - \Vert \bm{{\rm J}}_i(\eta)\bm{{\rm b}}_{i,a} \Vert^2_{\bm{{\rm W}}} \right] \Big\vert \leq 2D_{i,a}F_{i,a}+F_{i,a}^2.
\end{equation}
Finally, summing up all directions and using the triangle inequality yields
\begin{equation}\label{eq:central_difference_error_v}
	\setlength\abovedisplayskip{3pt}
	\setlength\belowdisplayskip{3pt} 
	\begin{aligned}
	\vert \overline{v_i} - \hat{v}_i \vert &= \frac{\sigma_g^2}{|\mathcal{A}|}\Big\vert \sum_{a\in\mathcal{A}} \mathbb{E}_{\eta}\left[\Vert \bm{{\rm J}}_i(\eta)\bm{{\rm b}}_{i,a} \Vert^2_{\bm{{\rm W}}} - \Vert \hat{\bm{{\rm d}}}_{i,a}(\eta) \Vert^2_{\bm{{\rm W}}} \right] \Big\vert \\
	& \leq \frac{\sigma_g^2}{|\mathcal{A}|}\sum_{a\in\mathcal{A}} \Big\vert \mathbb{E}_{\eta}\left[\Vert \bm{{\rm J}}_i(\eta)\bm{{\rm b}}_{i,a} \Vert^2_{\bm{{\rm W}}} - \Vert \hat{\bm{{\rm d}}}_{i,a}(\eta) \Vert^2_{\bm{{\rm W}}} \right] \Big\vert \\
	& \leq \frac{\sigma_g^2}{|\mathcal{A}|}\sum_{a\in\mathcal{A}}\left(2D_{i,a}F_{i,a} + F_{i,a}^2\right) \triangleq \epsilon_{i}^{{\rm fd}}.
\end{aligned}
\end{equation}
Further, based on (\ref{eq:central_difference_error_v}), the following result holds:
\begin{equation}
	\setlength\abovedisplayskip{3pt}
	\setlength\belowdisplayskip{3pt} 
	\Big\vert \sum_{i\in\mathcal{I}_c}z_i \overline{v_i} - \sum_{i\in\mathcal{I}_c}z_i \hat{v}_i \Big\vert \leq \sum_{i\in\mathcal{I}_c}z_i \vert \overline{v_i} - \hat{v}_i \vert \leq \sum_{i\in\mathcal{I}_c}z_i \epsilon_{i}^{{\rm fd}}.
\end{equation}
Consequently, we define the accumulated error bound related to central-difference approximation as
\begin{equation}
	\setlength\abovedisplayskip{3pt}
	\setlength\belowdisplayskip{3pt} 
	\epsilon_{{\rm fd}} \triangleq \max_{\bm{{\rm z}} \in \mathcal{Z}_c} \sum_{i\in\mathcal{I}_c}z_i \epsilon_{i}^{{\rm fd}}.
\end{equation}

\subsubsection{Candidate screening}
In our proposed method, screening restricts the feasible refinement decisions by excluding buildings outside $\mathcal{I}_c$. Define
\begin{equation}
	\setlength\abovedisplayskip{3pt}
	\setlength\belowdisplayskip{3pt}
	\epsilon_{{\rm scr}} \triangleq \max_{\bm{{\rm z}}\in\mathcal{Z}} \sum_{i\in \mathcal{I} \backslash \mathcal{I}_c}z_iv_i.
\end{equation}
This quantity bounds the total refinement value that could be omitted under the original budget.

\begin{proposition}
	Let $\bm{{\rm z}}^\star \in \operatorname*{arg\,min}_{\bm{{\rm z}}\in \mathcal{Z}}R(\bm{{\rm z}})$, and let $\hat{\bm{{\rm z}}}$ be an exact solution of (\textbf{P4}). We have
	\begin{equation}
		\setlength\abovedisplayskip{3pt}
		\setlength\belowdisplayskip{3pt} 
		0\leq R(\hat{\bm{{\rm z}}}) - R(\bm{{\rm z}}^\star) \leq 2\epsilon_{{\rm lin}} + 2\epsilon_{{\rm str}} + 2\epsilon_{{\rm fd}} + \epsilon_{{\rm scr}}.
	\end{equation}
\end{proposition}

\begin{IEEEproof}
	Please refer to Appendix \ref{proof_of_proposition1}.
\end{IEEEproof}

\subsection{Computational Complexity Analysis}
The overall computational cost of the proposed method consists of geometric screening, perturbation-based sensitivity estimation, and budget-constrained selection.

First, given the building proxies and Rx proxy set, evaluating the ellipsoid-building overlap requires computing the vertical overlap thickness at $Q_i$ footprint samples for each building-Rx pair. Each thickness evaluation involves solving a quadratic inequality for the vertical intersection with the ellipsoid and clipping the resulting interval to the building height, which takes constant time. The overlap computations therefore require $\mathcal{O}(N\sum_{i=1}^{N_b}Q_i)$ operations. Without spatial acceleration, identifying the primary LoS blocker requires testing each Tx–Rx segment against all building proxies, with a cost of $\mathcal{O}(N\sum_{i=1}^{N_b}H_i)$. Selecting the $M_c$ candidates by sorting their screening scores requires $\mathcal{O}(N_b{\rm log}N_b)$ operations. Consequently, the screening complexity is 
\begin{equation}
	\setlength\abovedisplayskip{3pt}
	\setlength\belowdisplayskip{3pt} 
	\mathcal{C}_{{\rm screen}} = \mathcal{O}\left(N\sum_{i=1}^{N_b}(Q_i+H_i) + N_b{\rm log}N_b \right),
\end{equation}
where $H_i$ denotes the number of half-space constraints defining the convex-prism proxy $\widetilde{B}_i$, including its vertical side planes and the top and bottom planes. In addition, building-proxy construction, footprint sampling, and Rx-proxy generation incur an additional pre-processing cost, denoted by $\mathcal{C}_{{\rm prep}}$, which can be reused when the underlying scene and proxy configuration remain unchanged.

For sensitivity estimation, the three transformation modes $\mathcal{A}=\{x, y, h\}$ require two perturbed response evaluations per mode. Thus, $6M_c$ RT evaluations are required. Let $\mathcal{C}_{{\rm RT}}$ denote the representative cost of one evaluation, including the required scene update and RT computation, and let $\mathcal{C}_W$ denote the cost of forming a response difference and evaluating its weighted squared norm. The cost of sensitivity estimation is then
\begin{equation}
	\setlength\abovedisplayskip{3pt}
	\setlength\belowdisplayskip{3pt} 
	\mathcal{C}_{{\rm sens}} = \mathcal{O}(6M_c\mathcal{C}_{{\rm RT}} + 3M_c\mathcal{C}_W).
\end{equation} 

For equal refinement costs, problem (\textbf{P4}) is solved by selecting the largest ${\rm min}\{M_c, \lfloor B/c \rfloor \}$ sensitivity scores, requiring $\mathcal{O}(M_c{\rm log}M_c)$ operations using sorting. For heterogeneous integer-valued costs and budget, dynamic programming requires $\mathcal{O}(M_cB)$ operations, where $B$ is expressed in integer cost units. For general costs, the complexity depends on the mixed-integer solver, and an exact solution does not admit a general polynomial-time guarantee. Denoting the corresponding selection cost by $\mathcal{C}_{sel}$, the overall computational cost is $\mathcal{C}_{{\rm total}} = \mathcal{C}_{{\rm prep}} + \mathcal{C}_{{\rm screen}} + \mathcal{C}_{{\rm sens}} + \mathcal{C}_{sel}$.

Without ellipsoid-guided screening, the same directional sensitivity estimator would require $6N_b$ response evaluations. The proposed screening stage therefore reduces their number by a fraction of $(1-M_c / N_b)$.

\section{Numerical Results}\label{sec:numerical_results}
In this section, we evaluate the proposed algorithm in realistic dense urban WDT scenarios. We first describe the simulation setup and benchmark refinement strategies. Then, we compare different methods in terms of radio-map fidelity. Finally, we conduct an ablation study to investigate the impact of the hyperparameters.

\subsection{Simulation Setup}
We evaluate the proposed algorithm in three urban scenarios, denoted by ``11-SW-8C'', ``11-SW-14B'', and ``11-SW-7B'', which are constructed from the 3D spatial data released by the Lands Department of the Hong Kong SAR Government through the Open3Dhk platform~\cite{open3dhk}. As shown in Fig. \ref{fig:scenarios}, the three scenarios contain 728, 735, and 818 buildings, respectively, and exhibit irregular footprints, diverse building heights, and complex street layouts. The original city geometry in each scenario serves as a high-fidelity reference for evaluating radio map predictions.

\begin{figure}[t]
	\centering
	\includegraphics[width=3.3in]{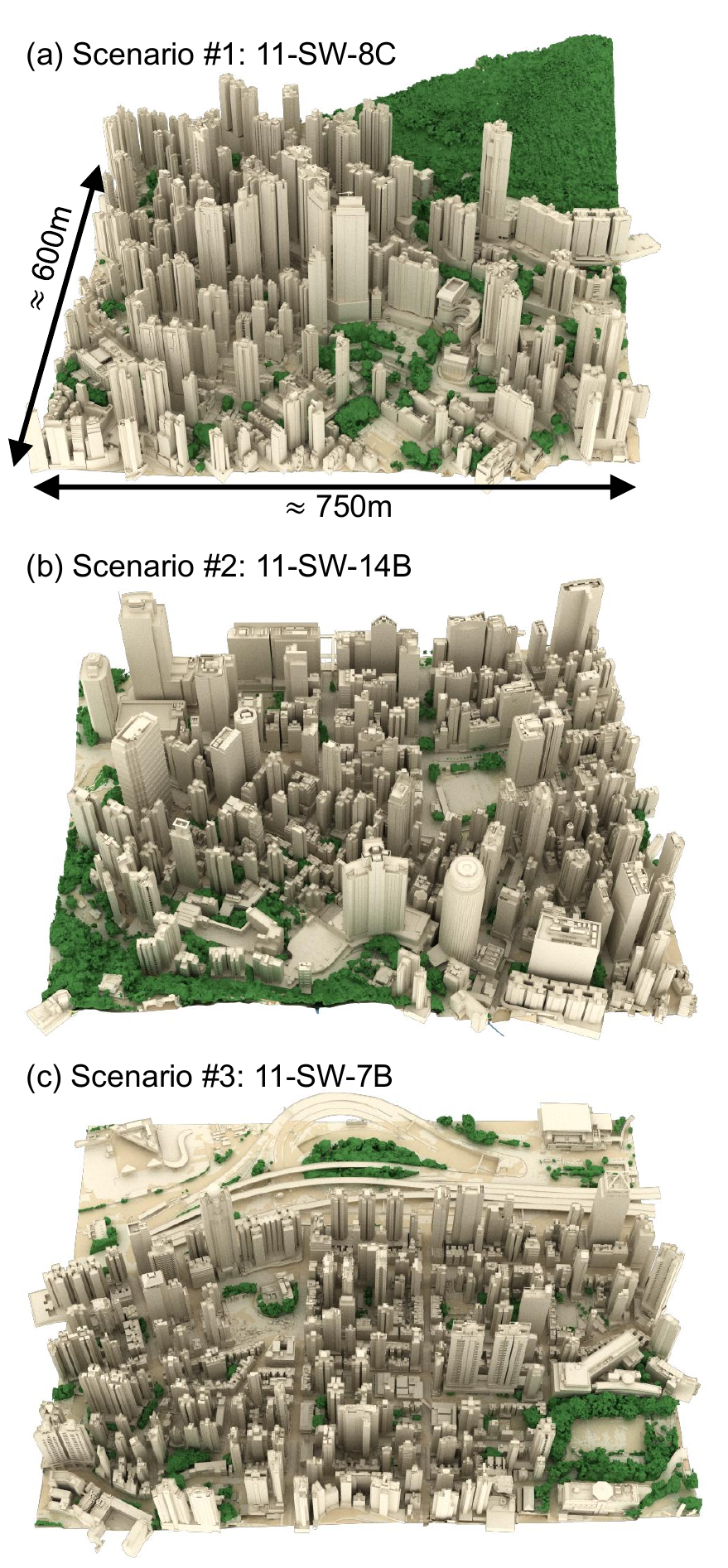}
	\caption{The considered urban scenarios.}
	\label{fig:scenarios}
\end{figure}

To emulate practical WDT construction under limited sensing and reconstruction resources, we sample the reference geometry at an initial point cloud density of 0.01 points/m$^2$ and apply Poisson surface reconstruction~\cite{Kazhdan2013screened} to obtain low-fidelity building meshes. When a building is selected for refinement, its point-cloud density is increased to 10 points/m$^2$ and its mesh is reconstructed using the same procedure. Buildings that are not selected retain their initial meshes. It should be noted that the reference geometry is used to generate point clouds, but is not available to the proposed algorithm when it estimates the building's refinement value. Material assignments and ray-tracing settings remain fixed throughout the refinement process. To reduce sensitivity to a particular Tx location, five feasible Tx deployments are considered in each scenario, including spatially separated outdoor and rooftop locations where applicable. The algorithms are run independently for each Tx deployment. Unless stated otherwise, the reported performance at each budget is averaged over the five deployments. The evaluation domains contain only valid outdoor Rx locations. They are generated using a terrain-aligned grid with a horizontal cell size of 5 m and an Rx height of 1.5 m above the terrain. Grid cells that intersect a building or lack valid terrain support are excluded. In this work, we consider radio map prediction and generate each radio map using Sionna RT at a carrier frequency of 3.5 GHz. Each Tx employs a $1\times 64$ vertically polarized array with half-wavelength element spacing and the antenna-element pattern specified in 3GPP TR 38.901. LoS and specular reflection paths are enabled with a maximum interaction depth of 3. Each radio map is computed using $10^7$ ray samples per Tx. The task response is the path gain in dB, with values floored at -120 dB. The radio map fidelity is measured by the RMSE between the refined and reference radio maps. The main hyperparameters are summarized in Table \ref{tab:parameters}.

\begin{table}[t]
	\centering
	\caption{Parameters Setup}
	\label{tab:parameters}
	\setlength{\tabcolsep}{5pt}
	\renewcommand{\arraystretch}{1.08}
	\begin{tabular}{l c}
		\toprule
		Parameter & Value \\
		\midrule
		
		Initial point-cloud density
		& $0.01$ points/m$^2$ \\
		
		Refined point-cloud density
		& $10$ points/m$^2$ \\
		
		Reconstruction method
		& Screened Poisson~\cite{Kazhdan2013screened} \\
		
		RT simulator
		& Sionna RT 1.2.1~\cite{Aoudia2025sionna} \\
		
		Carrier frequency
		& $3.5$ GHz~\cite{Shah2026vision} \\
		
		Tx array
		& $1\times64$, V-pol., $0.5\lambda$ spacing~\cite{Shah2026vision}\\
		
		Antenna-element pattern
		& 3GPP TR 38.901~\cite{3GPP38901} \\
		
		Propagation mechanisms
		& LoS and specular reflection \\
		
		Maximum interaction depth
		& $3$~\cite{Liu2026clutter}\\
		
		Ray samples per Tx
		& $10^7$~\cite{Iye2025open} \\
		
		Evaluation-grid spacing
		& $5$ m \\
		
		Rx height offset $h_r$
		& $1.5$ m~\cite{3GPP38901} \\
		
		Path-gain floor
		& $-120$ dB~\cite{Ye2026symNet} \\
		
		Tx deployments per scenario
		& $5$ \\
		
		Candidate-set size $M_c$
		& $75$ \\
		
		Excess path length $\Delta$
		& $50$ m \\
		
		Perturbation steps $(\delta_x,\delta_y,\delta_h)$
		& $(1,1,1)$ m \\
		
		Radial band boundaries
		& $\{50,100,200,350\}$ m \\
		
		Proxy spacing per radial band
		& $\{0.55,0.75,1.5,2.7,4.0\}$ m \\
		
		Minimum sectors per ring
		& $8$ \\
		
		LoS boost factor $\eta_{\rm LoS}$
		& $2$ \\
		
		Footprint sampling interval
		& $2$ m \\
		
		\bottomrule
	\end{tabular}
\end{table}

\subsection{Baselines}
For a fair comparison, all refinement methods operate on the same initial low-fidelity WDT and select and refine buildings under the same refinement budget and reconstruction algorithm. The considered baselines are summarized as follows.
\begin{itemize}
	\item \textbf{LoS blocker frequency-guided:} The score of each building is calculated as the frequency with which its geometric proxy blocks the LoS links between the Tx and the considered Rx locations.
	
	\item \textbf{Fresnel zone-guided:} The first Fresnel zone is constructed around the direct propagation path for each Tx--Rx pair, based on the carrier frequency. Buildings are ranked according to how extensively they overlap with these zones across the considered Rx region, with larger overlap indicating greater potential influence on signal propagation.
	
	\item \textbf{Distance-guided:} In this method, buildings are ranked according to their horizontal distances from the Tx. Buildings closer to the Tx are given higher refinement priorities.
	
	\item \textbf{Volume-guided:} Buildings are ranked according to their estimated geometric volumes, and those with the larger volumes are given higher refinement priorities.
	
	\item \textbf{Random:} Buildings are selected uniformly at random from the complete building set.
\end{itemize}

\subsection{Wireless Fidelity for Radio Map Computation}
\begin{figure}[t]
	\centering
	\includegraphics[width=3.3in]{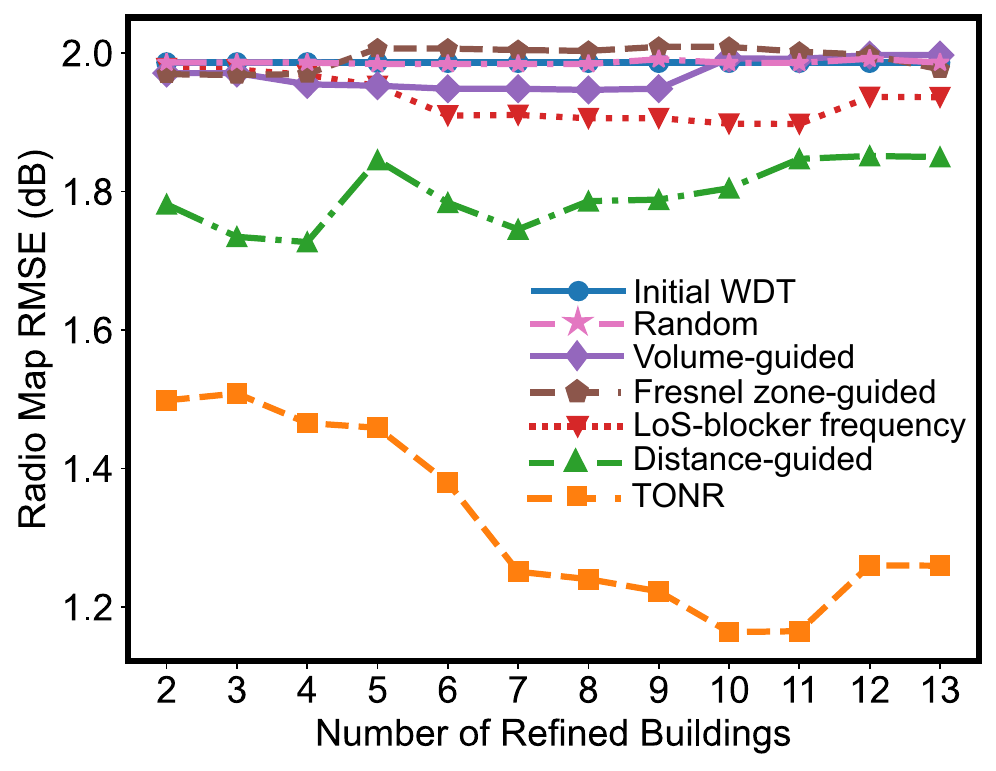}
	\caption{Averaged radio map RMSE under equal refinement costs in scenario \#1, 11-SW-8C.}
	\label{fig:11-SW-8C_equalcost}
\end{figure}

\begin{figure}[t]
	\centering
	\includegraphics[width=3.3in]{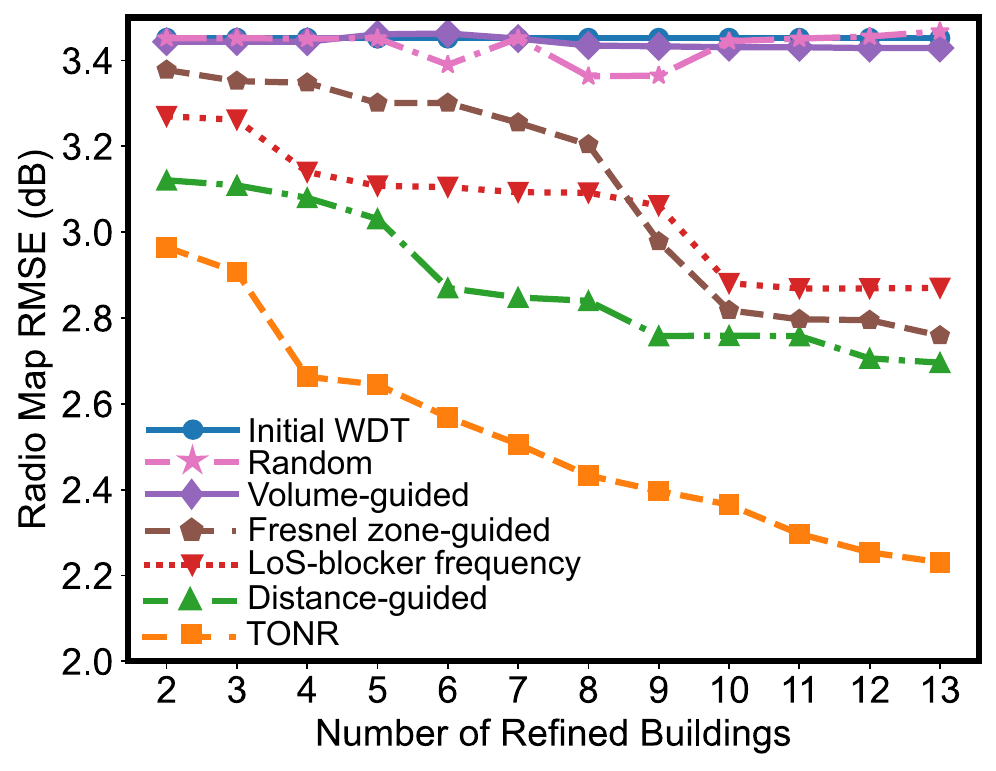}
	\caption{Averaged radio map RMSE under equal refinement costs in scenario \#2, 11-SW-14B.}
	\label{fig:11-SW-14B_equalcost}
\end{figure}

\begin{figure}[t]
	\centering
	\includegraphics[width=3.3in]{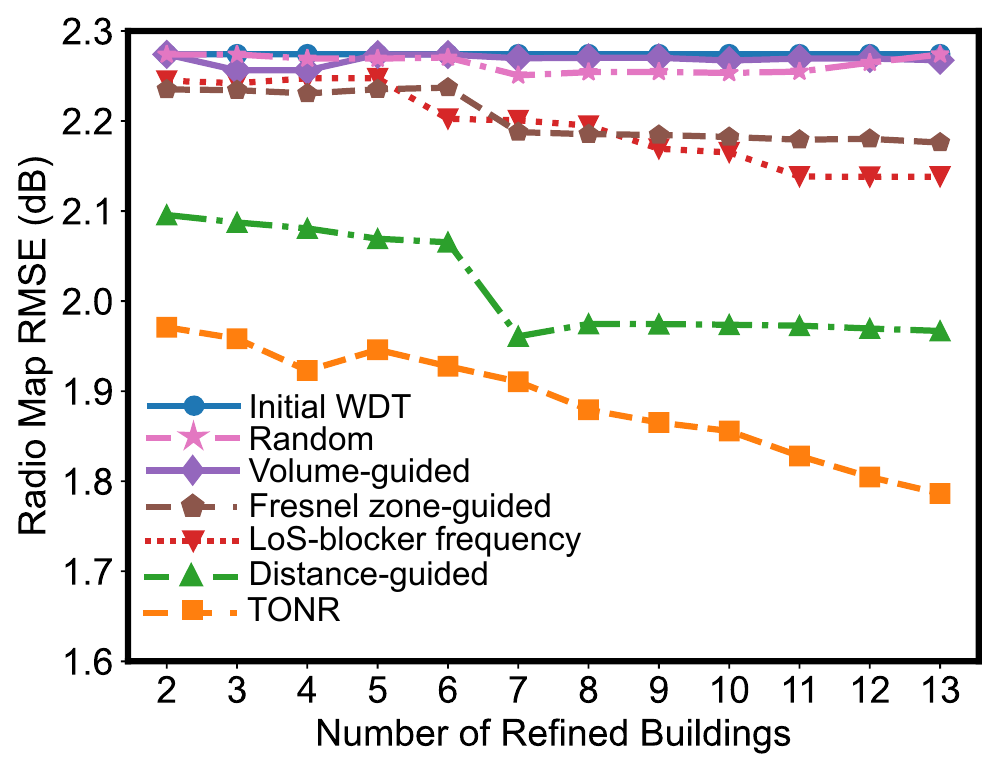}
	\caption{Averaged radio map RMSE under equal refinement costs in scenario \#3, 11-SW-7B.}
	\label{fig:11-SW-7B_equalcost}
\end{figure}

\begin{figure}[t]
	\centering
	\includegraphics[width=3.3in]{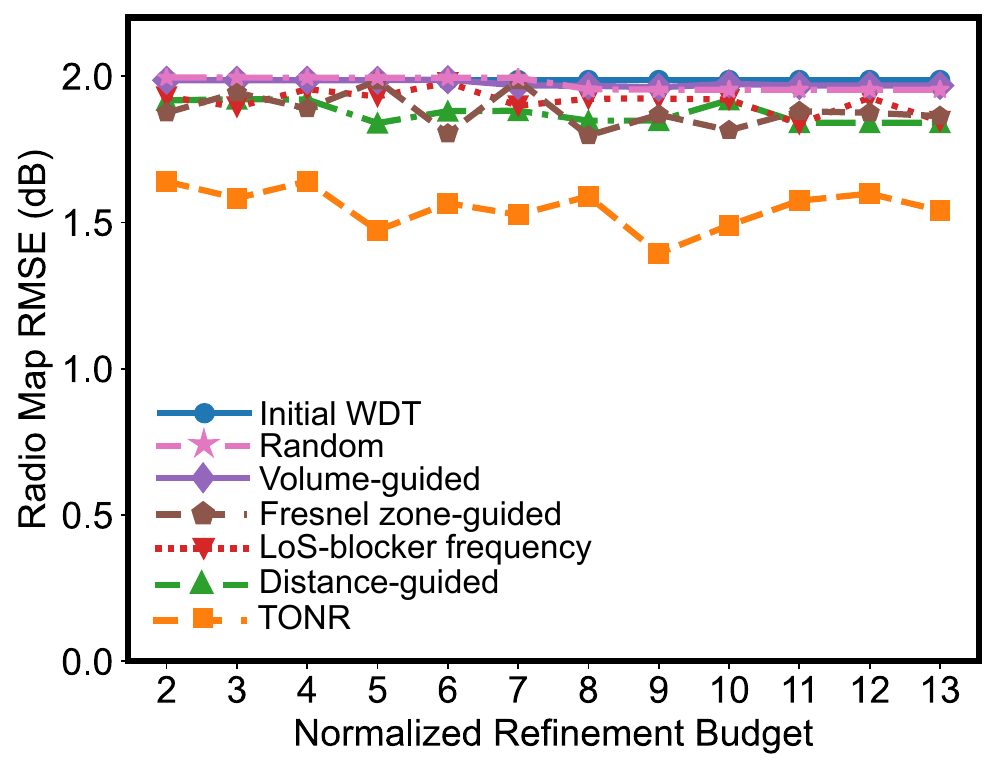}
	\caption{Averaged radio map RMSE under heterogeneous refinement costs in scenario \#1, 11-SW-8C.}
	\label{fig:11-SW-8C_differentcost}
\end{figure}

\begin{figure}[t]
	\centering
	\includegraphics[width=3.3in]{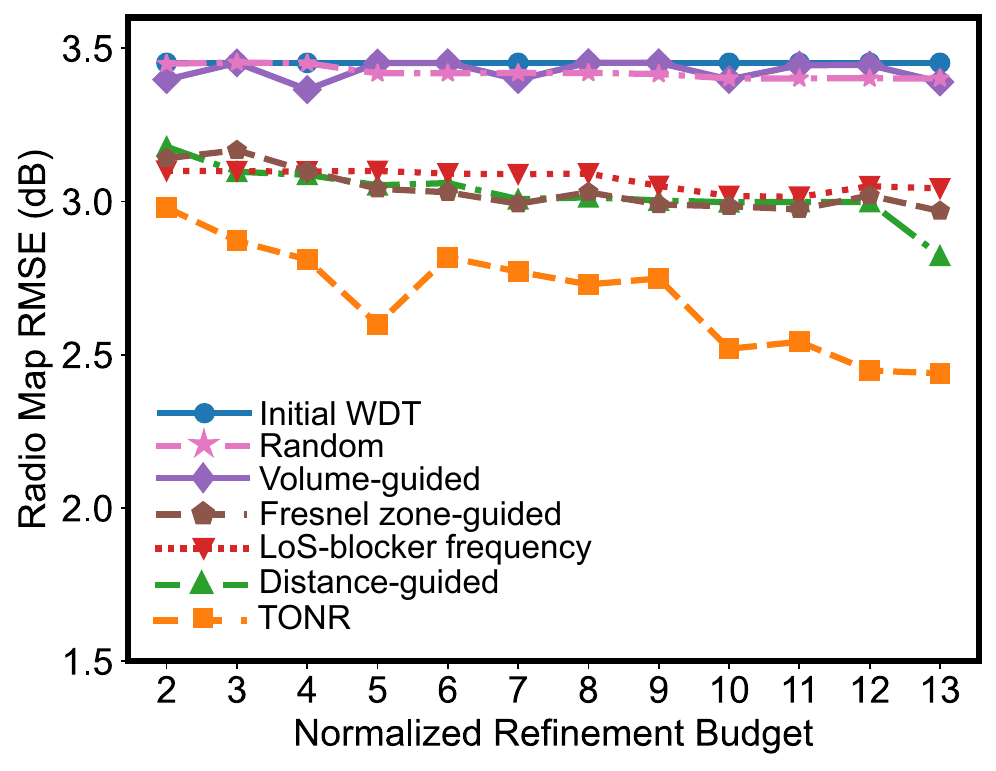}
	\caption{Averaged radio map RMSE under heterogeneous refinement costs in scenario \#2, 11-SW-14B.}
	\label{fig:11-SW-14B_differentcost}
\end{figure}

\begin{figure}[t]
	\centering
	\includegraphics[width=3.3in]{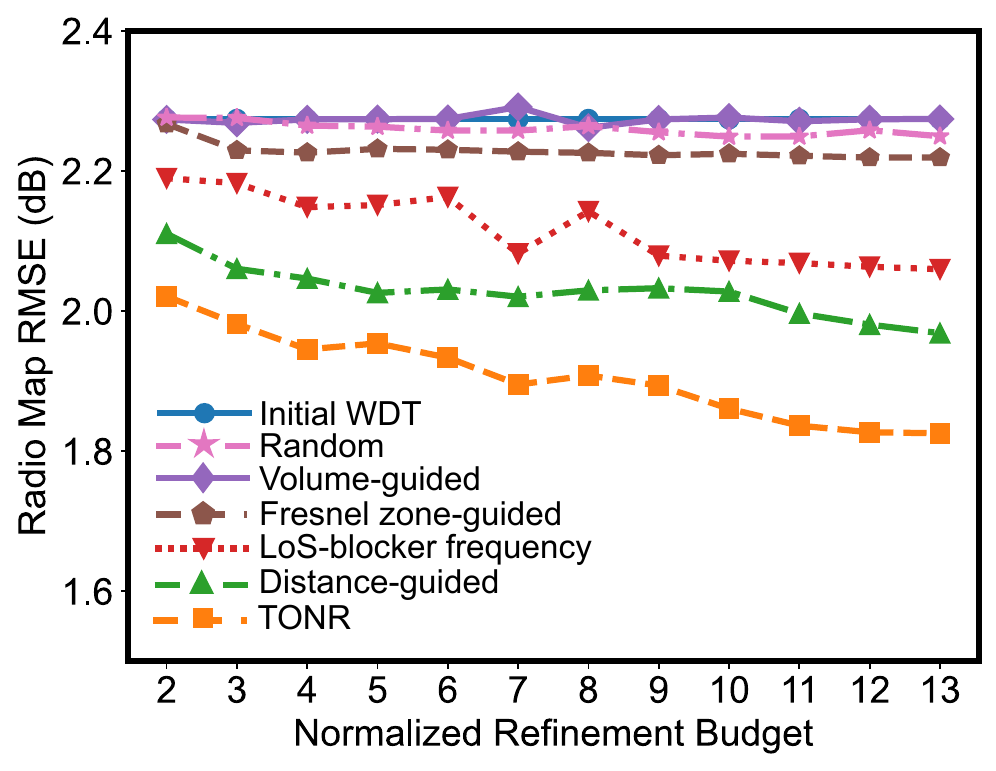}
	\caption{Averaged radio map RMSE under heterogeneous refinement costs in scenario \#3, 11-SW-7B.}
	\label{fig:11-SW-7B_differentcost}
\end{figure}

Figs. \ref{fig:11-SW-8C_equalcost}--\ref{fig:11-SW-7B_equalcost} compare the accuracy of radio map prediction under equal refinement costs. Each point reports the RMSE averaged over the five Tx deployments, and the budget directly specifies the number of refined buildings. First, it can be observed that the proposed method achieves the lowest average RMSE for every evaluated budget $K=2,\cdots,13$ in all three scenarios. Even with only two refined buildings, it reduces the RMSE of the initial WDT by 24.6\%, 14.1\%, and 13.3\% in scenarios \#1, \#2, and \#3, respectively. At $K=13$, the corresponding reductions reach 36.6\%, 35.4\%, and 21.4\%. For the considered baselines, random and volume-guided refinement generally remain close to the initial WDT because building size alone does not indicate whether a geometry error affects the radio map. LoS-blocking frequency and Fresnel-zone overlap provide geometric indicators of propagation relevance. However, these criteria do not directly quantify the expected improvement in radio-map fidelity resulting from building refinement. The distance-guided method is the strongest baseline in most cases, yet it remains consistently inferior to the proposed algorithm because proximity to the Tx does not necessarily imply a large task-level refinement value. 

Figs. \ref{fig:11-SW-8C_differentcost}--\ref{fig:11-SW-7B_differentcost} evaluate the more practical case of heterogeneous refinement costs. Since every selected building is reconstructed at the same target point cloud density, its additional point count is approximately proportional to its exterior area. We therefore estimate each cost from the exterior area of the initial WDT and normalize it by the scenario-wise average. Consequently, a budget of $B$ represents the estimated reconstruction cost of approximately $B$ average-sized buildings, rather than a fixed number of refined buildings. The proposed method again achieves the lowest average RMSE at every evaluated budget $B=2,\cdots,13$ in all three scenarios. At $B=2$, it reduces the RMSE of the initial WDT by 17.4\%, 13.7\%, and 11.1\% in scenarios \#1, \#2, and \#3, respectively. At $B=13$, the corresponding reductions increase to 22.4\%, 29.3\%, and 19.7\%. Compared with the best-performing baseline at $B=13$, the corresponding RMSE reductions are 16.3\%, 13.6\%, and 7.3\%. These results show that the proposed refinement value estimate remains effective when influential buildings have unequal refinement costs and must be selected through a knapsack constraint.
	
Another interesting observation is that although the proposed method is uniformly best among the evaluated strategies, its RMSE does not decrease monotonically at every budget. This is mainly because the proposed method assigns each building an individual refinement value and optimizes the sum of these values. However, the actual radio map RMSE is obtained after all selected buildings are jointly refined, and their propagation effects may interact nonlinearly. Moreover, under heterogeneous costs, increasing the budget may lead the knapsack solver to replace the previously selected buildings with a different combination rather than simply adding a new building. Consequently, a larger budget does not necessarily produce a lower measured RMSE at every adjacent budget point. However, the proposed method still maintains a clear overall advantage across Tx deployments, scenarios, and cost models, confirming that task-oriented refinement allocates limited reconstruction resources more effectively than purely geometry- or propagation-based heuristics. Future work will consider the interactions between different buildings and environmental components.

\subsection{Ablation Study}
We next examine how the three main design parameters affect the performance of the proposed TONR method. These parameters govern different stages of the refinement pipeline: the excess-path threshold $\Delta$ determines the spatial extent of the propagation-relevance ellipsoids, the candidate count $M_c$ controls how many screened buildings are retained for sensitivity evaluation, and the perturbation step $\delta$ determines the scale at which the central-difference sensitivities are estimated. Unless otherwise specified, we vary one parameter at a time while fixing the remaining parameters to the default configuration, i.e., $\Delta=50$ m, $M_c=75$, and $\delta=1$ m. The same configuration is used for all scenarios and Tx deployments without site-specific tuning. We consider equal-cost refinement budgets $K=2,\ldots,8$, with the point-cloud density of each selected building increased from 0.01 to 10 points/m$^2$.  Table \ref{tab:hyperparameter_ablation} reports the percentage reduction in radio-map RMSE, averaged first over the five Tx deployments and the considered budgets within each scenario, and then across the three scenarios.

\begin{table*}[t]
	\centering
	\caption{Hyperparameter ablation. Each entry gives the radio-map RMSE reduction (\%) averaged over five Tx deployments and $K=2,\ldots,8$. Higher values are better.}
	\label{tab:hyperparameter_ablation}
	\setlength{\tabcolsep}{10pt}
	\renewcommand{\arraystretch}{1.08}
	\begin{tabular}{l c c c c c}
		\toprule
		Parameter & Value & 11-SW-8C & 11-SW-14B & 11-SW-7B
		& Three-scenario mean \\
		\midrule
		
		\multirow{6}{*}{Excess path $\Delta$}
		& 5 m   & 29.10 & 21.37 & 15.30 & 21.93 \\
		& 10 m  & 29.58 & 22.18 & 15.29 & 22.35 \\
		& 30 m  & 29.49 & 22.66 & 15.09 & 22.41 \\
		& 50 m  & 29.49 & 22.66 & 15.09 & 22.41 \\
		& 100 m & 28.55 & 22.91 & 17.00 & \textbf{22.82} \\
		& 150 m & 25.47 & 17.56 & 13.46 & 18.83 \\
		\midrule
		
		\multirow{4}{*}{Candidate count $M_c$}
		& 25  & 19.15 & 24.45 & 12.73 & 18.78 \\
		& 50  & 29.02 & 23.37 & 13.78 & 22.06 \\
		& 75  & 29.49 & 22.66 & 15.09 & 22.41 \\
		& 100 & 29.49 & 22.18 & 18.60 & \textbf{23.42} \\
		\midrule
		
		\multirow{3}{*}{Perturbation step $\delta$}
		& 0.5 m & 28.68 & 22.04 & 14.06 & 21.59 \\
		& 1 m   & 29.49 & 22.66 & 15.09 & \textbf{22.41} \\
		& 2 m   & 30.28 & 20.82 & 14.86 & 21.99 \\
		
		\bottomrule
	\end{tabular}
\end{table*}

The excess-path threshold $\Delta$ controls the breadth of the geometric screening region and therefore balances propagation-path coverage against screening selectivity. As shown in Table \ref{tab:hyperparameter_ablation}, the proposed method is stable for $\Delta$ between 10 and 100 m, over which the three-scenario mean RMSE reduction varies only from 22.35\% to 22.82\%. This plateau indicates that the screening stage consistently retains the most relevant buildings over a relatively broad range of ellipsoid sizes. Although $\Delta=100$ m attains the highest mean reduction, it improves upon the default $\Delta=50$ m by only 0.41 percentage points. When $\Delta$ is further increased to 150 m, however, the mean reduction drops to 18.83\%. An overly large ellipsoid overlaps with many buildings that are geometrically admissible but weakly related to the dominant propagation paths. Because only $M_c$ buildings are retained, these less informative candidates can dilute the screening ranking and displace more influential buildings. The default value $\Delta=50$ m therefore provides a robust operating point without requiring the propagation-relevance region to be tuned to each scenario.

The candidate count $M_c$ controls the computational budget allocated to task-oriented sensitivity evaluation rather than serving as a scenario-specific parameter to be optimized. When $M_c$ is increased from 25 to 50, the three-scenario mean RMSE reduction improves from 18.78\% to 22.06\%, indicating that an overly aggressive screening configuration may exclude influential buildings before their task-level sensitivities are evaluated. Once $M_c$ reaches a moderate value, however, the average performance becomes relatively stable: the three-scenario mean varies only from 22.06\% to 23.42\% over $M_c=50$--$100$. In particular, increasing $M_c$ from 75 to 100 requires 33.3\% more perturbation-based RT evaluations, while improving the mean RMSE reduction by only 1.01 percentage points. We therefore use $M_c=75$ for all scenarios and Tx deployments as a fixed performance--complexity trade-off, without scenario-specific tuning. More generally, because the sensitivity-estimation stage requires $6M_c$ RT evaluations, $M_c$ can be selected directly according to the available offline computational budget.

Finally, the perturbation step $\delta$ determines the scale over which the task-response sensitivity is estimated. The three-scenario mean RMSE reduction changes by only 0.82 percentage points across $\delta\in\{0.5,1,2\}$ m, demonstrating that the proposed sensitivity estimator is not strongly dependent on this parameter within the tested range. A very small step provides a more local approximation but makes the response difference more susceptible to small numerical or simulation fluctuations, whereas a larger step produces a stronger response difference at the cost of deviating further from the local linear regime underlying the central-difference estimator. The intermediate choice $\delta=1$ m provides the best overall result, achieving a mean RMSE reduction of 22.41\%. Although $\delta=2$ m performs slightly better in 11-SW-8C, its performance decreases in the other two scenarios, particularly in 11-SW-14B. Thus, $\delta=1$ m offers the most consistent cross-scenario behavior.

Overall, the ablation results show that the proposed method does not rely on fine-grained scenario-specific parameter tuning. Moderate choices of $\Delta$ and $\delta$ provide stable performance, while $M_c$ offers an explicit and predictable trade-off between candidate coverage and offline RT cost. Accordingly, the fixed configuration $(\Delta,M_c,\delta)=(50~\mathrm{m},75,1~\mathrm{m})$ is used throughout all scenarios and Tx deployments, providing a favorable balance among wireless fidelity, cross-scenario robustness, and computational complexity.

\section{Conclusion}\label{sec:conclusion}
This paper investigated task-oriented nonuniform refinement (TONR) for WDTs under limited resources. We first derived an approximate building-level refinement value that captures the joint effects of correctable geometry uncertainty and the sensitivity of the task response to it. A structured perturbation model and central differences were used to estimate the surrogate refinement value from the initial low-fidelity WDT. Furthermore, a propagation-relevance ellipsoid-guided screening method was proposed to reduce the computational cost. We also derived an optimality-gap bound that accounts for modeling, estimation, and screening errors. Numerical results for radio map prediction across different realistic urban scenarios highlighted the value of selective refinement in WDTs, under both equal and heterogeneous refinement costs. By allocating the limited resources to buildings with high refinement values, the proposed method substantially improves wireless fidelity without refining the entire scene. Future work will account for interactions among jointly refined buildings.

\appendices
\section{Proof of Proposition 1} \label{proof_of_proposition1}
First, define
\begin{equation}
	\setlength\abovedisplayskip{3pt}
	\setlength\belowdisplayskip{3pt}
	V(\bm{{\rm z}}) \triangleq \sum_{i\in\mathcal{I}}z_iv_i, \quad \hat{V}(\bm{{\rm z}}) \triangleq \sum_{i\in\mathcal{I}}z_i\hat{v}_i.
\end{equation}
According to (\ref{eq:rz}), we have $\widetilde{R}(\bm{{\rm z}}) = \widetilde{R}(\bm{{\rm 0}}) - V(\bm{{\rm z}})$. Based on (\ref{eq:structured_error_v}) and (\ref{eq:central_difference_error_v}), we have
\begin{equation}
	\setlength\abovedisplayskip{3pt}
	\setlength\belowdisplayskip{3pt}
	\vert v_i - \hat{v}_i \vert \leq \epsilon_i^{{\rm str}} + \epsilon_i^{{\rm fd}}.
\end{equation}
Thus, for any $\bm{{\rm z}}\in\mathcal{Z}_c$, we can infer that
\begin{equation}
	\setlength\abovedisplayskip{3pt}
	\setlength\belowdisplayskip{3pt}
	\begin{aligned}
	\vert V(\bm{{\rm z}}) - \hat{V}(\bm{{\rm z}}) \vert &= \Big\vert \sum_{i\in \mathcal{I}_c}z_i (v_i - \hat{v}_i) \Big\vert \\
	& \leq \sum_{i\in \mathcal{I}_c}z_i \vert v_i - \hat{v}_i \vert \\
	& \leq \sum_{i\in \mathcal{I}_c}z_i\epsilon_i^{{\rm str}} + \sum_{i\in \mathcal{I}_c}z_i\epsilon_i^{{\rm fd}}.
\end{aligned}
\end{equation}
Based on the definitions of $\epsilon_{{\rm str}}$ and $\epsilon_{{\rm fd}}$, we obtain the following error bound on the total refinement value:
\begin{equation}\label{eq:error_bound_of_total_refinement_value}
	\setlength\abovedisplayskip{3pt}
	\setlength\belowdisplayskip{3pt}
	\vert V(\bm{{\rm z}}) - \hat{V}(\bm{{\rm z}}) \vert \leq \epsilon_{{\rm str}} + \epsilon_{{\rm fd}}, \forall \bm{{\rm z}}\in \mathcal{Z}_c.
\end{equation}

Because the optimal solution $\bm{{\rm z}}^\star$ of the original problem (\textbf{P1}) may select buildings outside the candidate set $\mathcal{I}_c$, we define
\begin{equation}
	\setlength\abovedisplayskip{3pt}
	\setlength\belowdisplayskip{3pt}
	z^\star_{c,i} = 
	\begin{cases}
		z_i^\star & i\in \mathcal{I}_c, \\
		0, & i\notin \mathcal{I}_c.
	\end{cases}
\end{equation}
It is obvious that 
\begin{equation}
	\setlength\abovedisplayskip{3pt}
	\setlength\belowdisplayskip{3pt}
	\sum_i c_iz^\star_{c,i} \leq \sum_i c_i z_i^\star \leq B.
\end{equation}
The total value lost by removing all buildings outside the candidate pool is
\begin{equation}
	\setlength\abovedisplayskip{3pt}
	\setlength\belowdisplayskip{3pt}
	V(\bm{{\rm z}}^\star) - V(\bm{{\rm z}}_c^\star) = \sum_{i\in \mathcal{I}\backslash\mathcal{I}_c} z_i^\star v_i \leq \max_{\bm{{\rm z}}\in\mathcal{Z}} \sum_{i\in \mathcal{I} \backslash \mathcal{I}_c}z_iv_i = \epsilon_{{\rm scr}}.
\end{equation}
In other words, $V(\bm{{\rm z}}^\star) \leq V(\bm{{\rm z}}_c^\star) + \epsilon_{{\rm scr}}$. Since both $\bm{{\rm z}}_c^\star$ and $\hat{\bm{{\rm z}}}$ belong to $\mathcal{Z}_c$, from the optimality of (\textbf{P4}), we know that $\hat{V}(\bm{{\rm z}}_c^\star) \leq \hat{V}(\hat{\bm{{\rm z}}})$. According to (\ref{eq:error_bound_of_total_refinement_value}), we can infer that
\begin{equation}
	\setlength\abovedisplayskip{3pt}
	\setlength\belowdisplayskip{3pt}
	\begin{aligned}
	V(\bm{{\rm z}}_c^\star) &\leq \hat{V}(\bm{{\rm z}}_c^\star) + \epsilon_{{\rm str}} + \epsilon_{{\rm fd}} \\
	&\leq \hat{V}(\hat{\bm{{\rm z}}}) + \epsilon_{{\rm str}} + \epsilon_{{\rm fd}} \\
	&\leq V(\hat{\bm{{\rm z}}}) + 2\epsilon_{{\rm str}} + 2\epsilon_{{\rm fd}},
	\end{aligned}
\end{equation}
that is, $V(\bm{{\rm z}}_c^\star) - V(\hat{\bm{{\rm z}}}) \leq 2\epsilon_{{\rm str}} + 2\epsilon_{{\rm fd}}$. Therefore, we have
\begin{equation}
	\setlength\abovedisplayskip{3pt}
	\setlength\belowdisplayskip{3pt}
	\begin{aligned}
	V(\bm{{\rm z}}^\star) - V(\hat{\bm{{\rm z}}}) &= [V(\bm{{\rm z}}^\star) - V(\bm{{\rm z}}_c^\star)] + [V(\bm{{\rm z}}_c^\star) - V(\hat{\bm{{\rm z}}})] \\
	& \leq \epsilon_{{\rm scr}} + 2\epsilon_{{\rm str}} + 2\epsilon_{{\rm fd}}.
	\end{aligned}
\end{equation}

Recall that in (\ref{eq:epsilon_lin}) we have $\vert R(\bm{{\rm z}}) - \widetilde{R}(\bm{{\rm z}}) \vert \leq \epsilon_{{\rm lin}}$ for any $\bm{{\rm z}}\in \mathcal{Z}$. Thus, for $\hat{\bm{{\rm z}}}$, we have $R(\hat{\bm{{\rm z}}}) - \widetilde{R}(\hat{\bm{{\rm z}}}) \leq \epsilon_{{\rm lin}}$, and for $\bm{{\rm z}}^\star$, we have $R(\bm{{\rm z}}^\star) \geq \widetilde{R}(\bm{{\rm z}}^\star) - \epsilon_{{\rm lin}}$. Further, we can obtain that 
\begin{equation}
	\setlength\abovedisplayskip{3pt}
	\setlength\belowdisplayskip{3pt}
	R(\hat{\bm{{\rm z}}}) - R(\bm{{\rm z}}^\star) \leq \widetilde{R}(\hat{\bm{{\rm z}}}) - \widetilde{R}(\bm{{\rm z}}^\star) + 2\epsilon_{{\rm lin}}.
\end{equation}
Similarly, according to (\ref{eq:rz}), we know that 
\begin{equation}
	\setlength\abovedisplayskip{3pt}
	\setlength\belowdisplayskip{3pt}
	\widetilde{R}(\hat{\bm{{\rm z}}}) - \widetilde{R}(\bm{{\rm z}}^\star) = V(\bm{{\rm z}}^\star) - V(\hat{\bm{{\rm z}}}).
\end{equation}
Thus, 
\begin{equation}
	\setlength\abovedisplayskip{3pt}
	\setlength\belowdisplayskip{3pt}
	\begin{aligned}
		R(\hat{\bm{{\rm z}}}) - R(\bm{{\rm z}}^\star) &\leq V(\bm{{\rm z}}^\star) - V(\hat{\bm{{\rm z}}}) + 2\epsilon_{{\rm lin}} \\
		& \leq \epsilon_{{\rm scr}} + 2\epsilon_{{\rm str}} + 2\epsilon_{{\rm fd}} + 2\epsilon_{{\rm lin}}.
	\end{aligned}
\end{equation}
Since $\bm{{\rm z}}^\star$ is the optimal solution for the original problem (\textbf{P1}), $R(\bm{{\rm z}}^\star) \leq R(\hat{\bm{{\rm z}}})$. Therefore, we obtain
\begin{equation}
	\setlength\abovedisplayskip{3pt}
	\setlength\belowdisplayskip{3pt}
	0\leq R(\hat{\bm{{\rm z}}})-R(\bm{{\rm z}}^\star) \leq \epsilon_{{\rm scr}} + 2\epsilon_{{\rm str}} + 2\epsilon_{{\rm fd}} + 2\epsilon_{{\rm lin}}.
\end{equation}
The proof is thus completed.

\vfill

\end{document}